\documentclass[a4paper, draftcls, onecolumn, 12 pt]{IEEEtran}
\usepackage{amsmath, amsfonts, amssymb}
\usepackage{citesort}
\usepackage{subfigure}
\usepackage{graphicx}
\usepackage{epsfig, color}
\def\blfootnote{\xdef\@thefnmark{}\@footnotetext} 
\long\def\symbolfootnote[#1]#2{\begingroup%
\def\thefootnote{\fnsymbol{footnote}}\footnote[#1]{#2}\endgroup}

\newcommand{\notx}{N_t}
\newcommand{\txind}{t}
\newcommand{\norx}{N_r}
\newcommand{\rxind}{r}
\newcommand{\fb}{k}

\newcommand{\fbvec}{\boldsymbol{\fb}}

\newcommand{\prfb}{q}
\newcommand{\round}{\ell}
\newcommand{\noround}{L}

\newcommand{\block}{b}
\newcommand{\noblock}{B}
\newcommand{\Blength}{J}
\newcommand{\cuind}{j}
\newcommand{\conste}{\mathcal{X}}
\newcommand{\consize}{M}
\newcommand{\rate}{R}
\newcommand{\chmat}{\boldsymbol{H}}
\newcommand{\acchmat}[1]{\chmat_{\overline{1, #1}}}
\newcommand{\rxsig}{\boldsymbol{Y}}
\newcommand{\acrxsig}[1]{\rxsig_{\overline{1, #1}}}
\newcommand{\txsig}{\boldsymbol{X}}
\newcommand{\actxsig}[1]{\txsig_{\overline{1, #1}}}
\newcommand{\nosig}{\boldsymbol{W}}
\newcommand{\acnosig}[1]{\nosig_{\overline{1, #1}}}
\newcommand{\info}{I}
\newcommand{\infoga}{\info^{\rm ga}}
\newcommand{\stairthre}{\hat{\info}}
\newcommand{\IX}[1]{\info_{\conste}\left(#1\right)}
\newcommand{\acinfo}[1]{\info_{\overline{1, #1}}}
\newcommand{\fbthre}{\overline{\info}}
\newcommand{\complex}{\mathbb{C}}

\newcommand{\diag}{{\rm diag}}
\newcommand{\txpow}{P}
\newcommand{\bartxpow}{\hat{\txpow}}
\newcommand{\code}{\mathcal{C}}
\newcommand{\expectation}[2]{\mathbb{E}_{#1}\left[#2\right]}
\newcommand{\nolevel}{K}
\newcommand{\pout}{p}
\newcommand{\upout}{\hat{\pout}}
\newcommand{\myPr}[1]{\Pr\left\{#1\right\}}
\newcommand{\diver}{d}
\newcommand{\Kdiverint}{\delta}

\newcommand{\prober}{P_e}
\newcommand{\rprober}{\prober^{(r)}}
\newcommand{\PEP}{P_{\rm PEP}}
\newcommand{\rdiver}{\diver^{\ddag}}
\newcommand{\tdiver}{\diver^{\dag}}
\newcommand{\achievediver}{\diver^{(r)}}

\newcommand{\diveronebit}{\hat{\diver}}
\newcommand{\divershort}{\underline{\diver}}
\newcommand{\stairrate}{\tau}
\newcommand{\power}{P}
\newcommand{\stair}{t}
\newcommand{\dist}{g}

\newcommand{\setopenr}[1]{\bigl[#1\bigr)}
\newcommand{\setopen}[1]{\left(#1\right)}

\newcommand{\goodset}{\mathcal{S}^{(\epsilon)}}
\newcommand{\openone}{\leavevmode\hbox{\small1\normalsize\kern-.33em1}}
\newcommand{\roundrun}{l}

\newcommand{\ifbthre}[1]{\info^{\dag}\left(#1\right)}
\newcommand{\inolevel}{\overline{\nolevel}}
\newcommand{\fbthresetr}[1]{\mathcal{A}_{\fbvec_{#1}}}
\newcommand{\fbthresetl}[1]{\overline{\mathcal{A}}_{\fbvec_{#1}}}
\newcommand{\indexset}[1]{\mathcal{S}_{\fbvec_{#1}}}

\newtheorem{theorem}{Theorem}

\newtheorem{remark}{Remark}

\title{MIMO ARQ with Multi-bit Feedback: Outage Analysis}

\author{Khoa D. Nguyen, Lars K. Rasmussen, Albert Guill\'{e}n i F\`{a}bregas \\and Nick Letzepis
\thanks{Khoa D. Nguyen and Nick Letzepis are with the Institute for Telecommunications Research, University of South Australia, Australia. {\tt email: dangkhoa.nguyen@postgrads.unisa.edu.au; nick.letzepis@unisa.edu.au}.}
\thanks{Lars K. Rasmussen was with the Institute for Telecommunications Research, University of South Australia, Australia. He is now with the Communication Theory Laboratory, Royal Institute of Technology, and the ACCESS Linneaus Centre, Stockholm, Sweden. {\tt email: lars.rasmussen@ee.kth.se}.}
\thanks{Albert Guill\'{e}n i F\`{a}bregas is with the Department of Engineering, University of Cambridge, Cambridge, UK. {\tt email: albert.guillen@eng.cam.ac.uk}.}
\thanks{This work has been presented in part at the Australian Communication Theory Workshop, Sydney, Australia,  4--7 Feb. 2009 and will be presented in part at the 2009 IEEE International Symposium on Information Theory, Seoul, South Korea, 28 June--3 July 2009.}
\thanks{This work was supported by the Australian Research Council under ARC grants RN0459498, DP0558861 and DP0881160.}}

\begin{document}
\maketitle
\begin{abstract}
We study the asymptotic outage performance of incremental redundancy automatic repeat request (INR-ARQ) transmission over the  multiple-input multiple-output (MIMO) block-fading channels with discrete input constellations.  We first show that transmission with random codes using  a discrete signal constellation across all transmit antennas achieves the optimal outage diversity given by the Singleton bound.  We then analyze the optimal SNR-exponent and outage diversity of INR-ARQ transmission over the MIMO block-fading channel.   We show that a significant gain in outage diversity is obtained by providing more than one bit feedback at each ARQ round.  Thus, the outage performance of INR-ARQ transmission can be remarkably improved with minimal additional overhead.  A suboptimal feedback and power adaptation rule, which achieves the optimal outage diversity, is proposed for MIMO INR-ARQ, demonstrating the benefits provided by multi-bit feedback. 
\end{abstract}
\newpage
\section{Introduction}
The block-fading channel \cite{OzarowShamaiWyner1994,BiglieriProakisShamai1998} is a useful mathematical model for many practical wireless communication scenarios.  The channel consists of a finite number of consecutive or parallel transmission blocks, where each block is affected by an independent fading coefficient.  The model well approximates the characteristics of delay-limited transmission over slowly varying channels, such as Orthogonal Frequency Division Multiplexing (OFDM) transmission over slowly-fading frequency-selective multipath channel, as well as narrow band transmission with frequency hopping as encountered in the Global System for Mobile Communications (GSM) and the Enhanced Data rate for  GSM evolution (EDGE) standards.  

Due to the finite number of fading blocks, the information rate supported by the channel depends on the instantaneous channel realization, and therefore is a random variable.  When the instantaneous mutual information is less than the transmission rate, transmission is in outage \cite{BiglieriProakisShamai1998}.  In this case, it follows from the strong converse theorem (see, e.g., \cite{Arimoto1973,MalkamakiLeib1999,CoverThomas2006}) that  messages are decoded in error with probability one \cite{CaireTuninetti2001,CaireTariccoBiglieri1999}.  Furthermore, it is shown in \cite{MalkamakiLeib1999,CaireTuninettiVerdu2004} that the use of sufficiently long-random codes achieves an average error rate equal to the outage probability.  Therefore, the outage probability is a fundamental limit on the performance of block-fading channels.

 Multi-input multi-output (MIMO) transmission has revolutionized modern wireless communications, and is now a key technology used in most current standards, e.g. WiFi (IEEE 802.11) and WiMax (IEEE 802.16) \cite{NandaWaltonKetchumWallaceHoward2005,GhoshWoltersAndrewsChen2005}.  Moreover, due to the randomness of the communication rate supported by the channel, it is essential to use adaptive techniques to enable high-rate reliable communication, where the transmission rate and/or power is adjusted to the channel realization.  The use of adaptive techniques depends strongly  on the availability of channel state information (CSI) at the transmitter and the receiver.  In most communication systems, CSI can be estimated at the receiver, while CSI is usually not directly available at the transmitter. The use of automatic-repeat-request (ARQ) transmission techniques is therefore a powerful approach for providing transmitter CSI, which in turn can be used to significantly improve the performance over block-fading channel \cite{CostelloHagenauerImaiWicker1998}. In this paper, we study the outage diversity of incremental redundancy (INR) ARQ, which characterizes the slope of the outage probability curve at high signal-to-noise ratio (SNR) in the log-log scale.

Before addressing the INR-ARQ case, we first consider fixed-rate transmission over the MIMO block-fading channel.  The optimal diversity-multiplexing tradeoff for a MIMO channel with optimal (Gaussian) input constellation has been characterized in \cite{ZhengTse2003}.  For systems with discrete input constellations, the rank criterion for the optimal outage diversity was derived in \cite{TarokhSeshadriCalderbank1998} from a worst-case analysis of the pair-wise error probability (PEP).  References \cite{LuKumar2003,LuKumar2005} establish the Singleton bound on the optimal SNR-exponent of quasi-static MIMO channels with discrete input constellations.  The Singleton bound is achievable by a wide range of input constellations via a unified code construction method proposed in \cite{LuKumar2005}.  In this paper, we show that the Singleton bound is the outage diversity, and is achievable by random codes with an arbitrary discrete input constellation.  This rigorously proves that the Singleton bound is the optimal SNR-exponent of MIMO transmission with discrete input constellations.  The result will also prove instrumental in designing and analyzing  INR-ARQ transmission over the MIMO block-fading channels.



As our main focus, we propose a multi-bit feedback INR-ARQ transmission scheme, and study its outage performance over the MIMO block-fading channel.  In an INR-ARQ scheme, transmission starts with a high-rate codeword, and additional redundancy is requested via a feedback link when the codeword is not successfully decoded.   Transmission is in outage if the codeword is not decodable within the maximum delay constraint allowed by the system.  Traditional INR-ARQ systems implement one-bit feedback from the receiver, indicating whether additional redundancy is required.  However, due to the accumulative nature of INR-ARQ schemes, performance improvements are possible when additional information regarding the status of the current transmission is provided through the feedback link.  Several multi-bit feedback INR-ARQ schemes have been proposed in the literature.  In \cite{VisotskySunTripathi2005,YiqingJiangzhou2006}, a multi-bit feedback INR-ARQ scheme was proposed for convolutional codes, while in \cite{SteinerShamai2008} a multi-layer broadcasting strategy with multi-bit feedback was shown to improve the throughput performance of INR-ARQ transmission.  There is, however, no unified approach for designing  multi-bit feedback INR-ARQ transmission.  In this paper, we take an information-theoretic approach to analyzing and designing schemes for multi-bit feedback INR-ARQ transmission, which provides significant improvement in outage performance over the MIMO block-fading channel.

  An important performance measure for INR-ARQ transmission in the MIMO block-fading channel is the rate-diversity-delay tradeoff.  This tradeoff has only been studied for INR-ARQ systems with one-bit ACK/NACK feedback in \cite{GamalCaireDamen2006,ChuangGuillenRasmussenCollings2008,NguyenRasmussenGuillenLetzepis2009}.  In particular, reference \cite{GamalCaireDamen2006} characterizes the rate-diversity-delay tradeoff of Gaussian input MIMO INR-ARQ systems with both constant and adaptive transmit power.  Reference \cite{ChuangGuillenRasmussenCollings2008} studies the tradeoff for MIMO INR-ARQ systems with rotated discrete input constellation, which increases the system diversity at the cost of increasing signal constellation size and hence the decoding complexity.  
 Inspired by \cite{ChuangGuillenRasmussenCollings2008}, the results of \cite{NguyenRasmussenGuillenLetzepis2009} show that power adaptation  offers significant gains in outage probability of INR-ARQ transmission.  In this paper, we extend the results of \cite{NguyenRasmussenGuillenLetzepis2009} to INR-ARQ systems with multi-bit feedback \cite{NguyenRasmussenGuillenLetzepis2009isit}.  In particular, we characterize the rate-diversity-delay tradeoff of multi-bit feedback MIMO INR-ARQ systems and show that multi-bit feedback and optimal power adaptation provide significant outage diversity gains for transmission over the block-fading channel.  A suboptimal feedback and power adaptive rule is also proposed, illustrating the benefits offered by multi-bit feedback.  

The remainder of the paper is organized as follows. Section \ref{sec:system-model} describes  the MIMO block-fading channel model.  Section \ref{sec:preliminaries} proposes the multi-bit feedback INR-ARQ system based on mutual information and information outage.  Sections \ref{sec:asymptotic-analysis} and \ref{sec:power-adapt-feedb} discuss system design and performance analysis. Finally, concluding remarks are given in Section \ref{sec:conclusion} and proofs are collected in the Appendices. 
\section{Channel Model}
\label{sec:system-model}
Consider INR-ARQ transmission over a MIMO block-fading channel with $\notx$ transmit and $\norx$ receive antennas.  Each ARQ round is transmitted over $\noblock$ additive white Gaussian noise (AWGN) blocks of $\Blength$ channel uses each, where block $\block$ at ARQ round $\round$ is affected by a flat fading channel gain matrix $\chmat_{\round,\block} \in \mathbb{C}^{\norx\times \notx}$.  The baseband equivalent of the channel in the $\round$-th ARQ round is given by 
\begin{equation}
  \label{eq:channel_model}
  \rxsig_{\round} = \sqrt{\frac{\txpow_{\round}}{\notx}}\chmat_{\round} \txsig_{\round} + \nosig_{\round}, 
\end{equation}
where $\txpow_{\round}$ is the transmit power in round $\round$, $\txsig_{\round} \in \complex^{\noblock \notx \times  \Blength}, \rxsig_{\round}, \nosig_{\round} \in \complex^{\noblock \norx \times  \Blength}$ are correspondingly the transmitted signal, the received signal, and the additive noise; while $\chmat_{\round} \in \complex^{\noblock\norx\times\noblock\notx}$ is  a  block diagonal channel gain matrix at round $\round$ with 
\[
\chmat_{\round}= \diag(\chmat_{\round, 1}, \ldots, \chmat_{\round, \noblock}). 
\]
In the INR-ARQ scheme, the receiver attempts to decode  at round $\round$ based on the  received signals collected  in rounds $1, \ldots, \round$.  The entire channel after  $\round$ ARQ rounds is
\begin{equation}
\label{eq:channel_model2}
\acrxsig{\round} =\acchmat{\round} \actxsig{\round}+ \acnosig{\round}, 
\end{equation}
where
\begin{align}
\acrxsig{\round}&= \left[\rxsig_{1}', \ldots, \rxsig_{\round}'\right]'\notag\\
\actxsig{\round}&=\left[\txsig_{1}', \ldots, \txsig_{\round}'\right]'\notag\\
\acchmat{\round}&=\diag\left( \sqrt{\frac{\txpow_{1}}{\notx}}\chmat_{1}, \ldots,  \sqrt{\frac{\txpow_{\round}}{\notx}}\chmat_{\round}\right)\notag\\
\acnosig{\round}&=\left[\nosig_{1}', \ldots, \nosig_{\round}'\right]' \notag
\end{align}
and $(\cdot)'$ denotes non-conjugate transpose. 
  We consider transmission where the entries of $\txsig_\round $ are drawn from an input constellation $\conste \subset \mathbb{C}$ of size $2^M$, and assume that the constellation $\conste$ has unit average energy, i.e.,  entries $x \in \conste$ of $\txsig_{\round}$ satisfy $\expectation{}{|x|^2} = 1$.  We further assume that the entries of $\chmat_{\round,\block}$ and $\nosig_{\round}$ are independently drawn from a unit-variance Gaussian complex distribution $\mathcal{N}_{\complex}(0, 1)$, and that $\chmat_{\round,\block}$ is available at the receiver.  The average SNR at each receive antenna is then $\txpow_{\round}$.  

We consider ARQ transmission with a long-term power constraint, where the average  power for each codeword is constrained to $\txpow$, i.e., 
\begin{equation}
\label{eq:expect_pow_constraint}
\expectation{\acchmat{\noround}}{\sum_{\round=1}^{\noround}\txpow_{\round}} \leq \txpow, 
\end{equation}
where  $\txpow_{\round}$ is adapted to $\acchmat{\round-1}$ through receiver feedback.  
\section{Preliminaries}
\label{sec:preliminaries}
\subsection{Accumulated Mutual Information}
Assuming that the realized channel matrix at round $\round$ is $\chmat_{\round}$, the input-output mutual information of the MIMO channel in round $\round$ is 
\begin{equation}
\label{eq:info_1round}
\info_{\round}\left(\sqrt{\frac{\txpow_{\round}}{\notx}}\chmat_{\round}\right) = \frac{1}{\noblock}\sum_{\block=1}^{\noblock}\IX{\sqrt{\frac{\txpow_{\round}}{\notx}} \chmat_{\round,\block}}, 
\end{equation}
where $\IX{\sqrt{\frac{\txpow_{\round}}{\notx}} \chmat_{\round,\block}}$ is the input-output mutual information \cite{CoverThomas2006}, measured in bits per channel use (bpcu), of an AWGN MIMO channel with input constellation $\conste$ and channel  matrix $\sqrt{\frac{\txpow_{\round}}{\notx}}\chmat_{\round,\block}$.  The average input-output mutual information after $\round$ ARQ rounds is given by $\frac{1}{\round}\sum_{\roundrun=1}^{\round}\info_{\roundrun}$ bpcu.  Let 
\begin{equation}
\label{eq:accinfo}
\acinfo{\round}\triangleq \sum_{\roundrun=1}^{\round} \info_{\roundrun}
\end{equation}
be the {\em accumulated mutual information} after $\round$ ARQ rounds.  We now propose the multi-bit feedback INR-ARQ transmission scheme based on the accumulated mutual information $\acinfo{\round}$. 


\subsection{Multi-Level Feedback}
\label{sec:inr-arq-system}
We consider an INR-ARQ system with a delay constraint of  $\noround$ ARQ rounds, where a feedback index $\fb \in \{0, \ldots, \nolevel-1\}$ is delivered after each transmission round through a zero-delay error-free feedback channel.  Power and rate adaptation are performed based on receiver feedbacks.  The overall system  model is illustrated in Figure \ref{fig:ARQ_sys}.  

\subsubsection{Transmitter}
\label{sec:transmitter}
Consider a code book $\code$ of rate $\frac{R_{\consize}}{\noround}$, $R_{\consize} \in (0, \consize)$ bits per coded symbol,  that maps a message $m \in \{1, \ldots, 2^{R_{\consize}\notx \noblock\Blength}\}$ to a codeword $\boldsymbol{x}(m) \in \conste^{\notx  \noblock \Blength \noround}$.  At transmission round $\round$,  $\notx \noblock \Blength$ of the  coded symbols are formatted into $\txsig_{\round}(m) \in \conste^{\noblock\notx \times \Blength}$ and transmitted via the channel in (\ref{eq:channel_model}) with power $\txpow_{\round}(\fbvec_{\round-1})$, where $\fbvec_{\round-1}= [\fb_1, \ldots, \fb_{\round-1}]$ is the vector of feedback indices collected from rounds $1, \ldots, \round-1$.  The realized code rate of a single ARQ round is $\rate \triangleq \rate_{\consize}\notx$ bpcu,  and the realized code rate after $\round$ ARQ rounds is $\frac{\rate}{\round}$ bpcu.  If feedback $\fb_{\round}= \nolevel-1$ (denoting positive acknowledgment (ACK)) is received  after $\round$ transmission rounds, the transmission is successful  and transmission of the next message starts.  Otherwise, the transmitter continues with new transmission rounds until feedback index $\nolevel-1$ is received or until $\noround$ transmission rounds have elapsed.  
\subsubsection{Receiver}
\label{sec:receiver}
Upon receiving round $\round$, the receiver attempts to decode the transmitted message from the received signals collected from rounds 1 to $\round$.  The receiver employs a decoder with error detection capabilities as described in \cite{CaireTuninetti2001}.  The decoder outputs $\hat{m} \in \{1, \ldots, 2^{\rate\noblock\Blength}\}$ if  there exists a unique message $\hat{m}$ such that $\actxsig{\round}(\hat{m})$ and $\acrxsig{\round}$ are jointly typical conditioned on $\acchmat{\round}$ \cite{CoverThomas2006}; then  an ACK is delivered to the transmitter via feedback index $\fb_{\round}= \nolevel-1$.  Otherwise, a quantization of the {\em accumulated mutual information} $\acinfo{\round}$ is delivered via feedback index $\fb_{\round}$ satisfying $\acinfo{\round} \in \setopenr{\ \fbthre([\fbvec_{\round-1}, \fb_{\round}]), \fbthre([\fbvec_{\round-1}, \fb_\round+1])}$, with predefined quantization thresholds $\fbthre(\fbvec_{\round}), \fbvec_{\round} \in \{0, \ldots, \nolevel-2\}^{\round}$, and $\fbthre([\fbvec_{\round-1},\nolevel-1])= \infty$ for $\round=1, \ldots, \noround-1$.  An example of the feedback thresholds for the first two rounds of an ARQ system with $\nolevel =4$ is illustrated in Figure \ref{fig:threshold}.   Feedback index $3$ is used to denote successful transmission.  At the first ARQ round, the leftmost set of feedback thresholds is used; while at the second ARQ round, one of the three sets of feedback thresholds on the right is employed, depending on which feedback index was delivered in the first round.  
Noting that $\acinfo{\round+1} \geq \acinfo{\round}$, the feedback thresholds in round $\round+1$ should be designed such that $\fbthre(\fbvec_{\round}) = \fbthre([\fbvec_{\round}, 0]) <  \ldots < \fbthre([\fbvec_{\round}, \nolevel-2])$.  Thus, the set of quantization thresholds is completely defined by $\fbthre(\fbvec_{\noround-1})$ for all practical purposes. 

\subsubsection{Power constraint}
The probability of having feedback vector $\fbvec_{\round}$ at round $\round$, denoted as $\prfb(\fbvec_{\round})$, is recursively expressed as
\begin{align}
\prfb(\fbvec_0)&= 1\\
\label{eq:prfb}\prfb([\fbvec_{\round-1}, \fb])&= \myPr{\fb_{\round}=\fb|\fbvec_{\round-1}}\prfb(\fbvec_{\round-1}),\\
 \myPr{\fb_{\round}= \fb\big| \fbvec_{\round-1}}&= 
\myPr{\acinfo{\round-1}+ \info_{\round} \in \setopenr{\ \fbthre([\fbvec_{\round-1}, k]), \fbthre([\fbvec_{\round-1}, k+1])}\big|\fbvec_{\round-1}}, \notag
\end{align}
where $\info_{\round}$ is given by (\ref{eq:info_1round}) with $\txpow_{\round}=\txpow_{\round}(\fbvec_{\round-1})$. 
Then, the power constraint in (\ref{eq:expect_pow_constraint}) can be written as
\begin{equation}
\label{eq:pow_constraint}
\expectation{\acchmat{\noround}}{\sum_{\round=1}^{\noround}\power_{\round}}=\txpow_1+ \sum_{\round=2}^{\noround}\sum_{\fbvec_{\round-1} \in \{0, \ldots, \nolevel-1\}^{\round-1}} \prfb(\fbvec_{\round-1}) \txpow_{\round}(\fbvec_{\round-1}) \leq \txpow.
\end{equation}

\subsection{Information Outage}
After $\round$ ARQ rounds, the input-output mutual information  is $\frac{\acinfo{\round}}{\round}$ and the realized code rate is $\frac{\rate_{\consize}\notx}{\round}= \frac{\rate}{\round}$.  Hence, transmission is in outage at round $\round$ if $\acinfo{\round} < \rate$.  The probability of having an outage at round $\round$ is then given by 
\begin{equation}
\label{eq:def_outage_prob}
\pout(\round) \triangleq \myPr{\acinfo{\round} < \rate}. 
\end{equation}
With an optimal coding scheme, and in the limit of the number of channel uses $\Blength \to \infty$, the codeword is correctly decoded whenever $\acinfo{\round} > \rate$; otherwise, an error is detected \cite{CaireTuninetti2001}.  Therefore, the outage probability $\pout(\round)$ is an achievable lower bound on the word error probability at round $\round$.  For INR-ARQ transmission with delay constraint $\noround$, the overall outage probability is $\pout(\noround)$.     


\section{Asymptotic  Analysis}
\label{sec:asymptotic-analysis}
 Consider a power adaptation rule $\txpow_{\round}= \txpow_{\round}(\fbvec_{\round-1})$ satisfying the power constraint in (\ref{eq:pow_constraint}).   We prove that for large  $\txpow$, the outage probability at round $\round$ behaves like
\begin{equation}
\pout(\round) \doteq \txpow^{-\diver_{\round}(R)}, 
\end{equation}
where $\diver_{\round}(R)$ is the outage diversity at round $\round$ and  the exponential equality ($\doteq$) indicates \cite{ZhengTse2003}
\begin{equation}
\label{eq:def_expeq}\diver_{\round} (R) = \lim_{\txpow \to \infty} \frac{-\log \pout(\round)}{\log \txpow}.
\end{equation}
We determine the optimal rate-diversity-delay tradeoff $\diver_{\round}(\rate)$ of ARQ systems with $\nolevel$ levels feedback and prove that the optimal outage diversity is achievable.
\subsection{MIMO Block-Fading without ARQ}
In order to characterize the outage diversity or achievable SNR-exponent for the MIMO INR-ARQ channel, we first need to study the corresponding limits for fixed-rate transmission over the MIMO block-fading channel. These results are keys to proving our main results for multi-bit ARQ.
\begin{theorem}
\label{the:asym_BF}
Consider fixed-rate transmission ($\noround =1$) with rate $\rate$ and power $\txpow$ over the MIMO block-fading channel in (\ref{eq:channel_model}) using constellation $\conste$ of size $2^{\consize}$ and the transmission scheme described in Section \ref{sec:inr-arq-system}.  Let $\info= \info_1\left(\sqrt{\frac{\txpow}{\notx}}\chmat_1\right)$ be the realized input-output  mutual information as defined in (\ref{eq:info_1round}).  For large $\txpow$, we have that 
\begin{align}
\label{eq:asymp_BF}
\myPr{\info < \rate} &\doteq \txpow^{-\diver(\rate)}, \\
\myPr{\info \leq \rate} &\doteq \txpow^{-\rdiver(\rate)},
\end{align}
where $\diver(\rate)$ is bounded by $\rdiver(\rate) \leq \diver(\rate) \leq \tdiver(\rate)$, and
\begin{align}
\label{eq:diver_BF_bound}
\tdiver(\rate) &\triangleq\norx\left(1+ \left\lfloor B\left(\notx-\frac{\rate}{\consize}\right)\right\rfloor\right)\\
\label{eq:diver_BF_achieve}\rdiver(\rate)&\triangleq \norx\left\lceil \noblock\left(\notx-\frac{\rate}{\consize}\right)\right\rceil.
\end{align}
Furthermore, $\rdiver(\rate)$ is the SNR-exponent for the case of using  random codes with rate $\rate$, where the code symbols are drawn  uniformly from $\conste$.  
\end{theorem}
\begin{proof}
See Appendix \ref{sec:proof-theor-refpr}.
\end{proof}

To the best of our knowledge, a rigorous proof of this result has not been reported in the literature.  The results of \cite{LuKumar2005} show that $\tdiver(\rate)$ is an upper bound to the optimal SNR-exponent, which is achievable with a range of fixed input constellations using the proposed code construction method.  Theorem \ref{the:asym_BF} shows that $\tdiver(\rate)$ is actually the outage diversity of the MIMO block-fading channel with an arbitrary input constellation of size $2^\consize$.  Furthermore, $\tdiver(\rate)$ is achievable by using random codes with an arbitrary fixed input constellation.


\subsection{Multi-bit MIMO ARQ}
We now consider  ARQ transmission over the block-fading channel in (\ref{eq:channel_model}) using input constellation $\conste$ as described in Section \ref{sec:transmitter}.   Using Theorem \ref{the:asym_BF}, the optimal rate-diversity-delay tradeoff of the MIMO INR-ARQ scheme with multi-bit feedback  is characterize as follows. 
\begin{theorem}
\label{the:out_diver_K}
Consider INR-ARQ transmission over the MIMO block-fading channel in (\ref{eq:channel_model}) using constellation $\conste$ of size $2^{\consize}$ and the transmission scheme described in Section \ref{sec:inr-arq-system}, where a codeword is considered successfully delivered at  round $\round$ if $\acinfo{\round} \geq \rate$.  Assume that the number of feedback levels is $\nolevel \geq \left\lceil\frac{\noblock \rate}{\consize}\right\rceil+1$.  Subject to the power constraint in (\ref{eq:pow_constraint}), the optimal rate-diversity-delay tradeoff is given by
\begin{equation}
\label{eq:out_diver_K} \diver_{\round}(\rate) =(1+\noblock\notx\norx)^{\round-1}\left(\tdiver(\rate)+1\right)-1
\end{equation}
when $\frac{\noblock\rate}{\consize}$ is not an integer, where $\tdiver(\rate)$ is given in Theorem \ref{the:asym_BF}.
\end{theorem}
\begin{proof}
See Appendix \ref{sec:proof-out-diver-K}.
\end{proof}
\begin{remark}
  \label{rem:set_thres}
  Beside the optimal outage diversity, the proof also shows the following result, which is useful in designing the feedback rules.
  \begin{itemize}
  \item The optimal outage diversity of INR-ARQ systems is achievable with $\left\lceil \frac{\noblock\rate}{\consize}\right\rceil+1$ feedback levels, where the feedback thresholds of each round are fixed at $\stairthre_{\stair}= \frac{\consize\stair}{\noblock}, \stair=0, \ldots, \left\lfloor\frac{\noblock\rate}{\consize}\right\rfloor$.    Therefore, for systems with $\nolevel \geq \left\lceil \frac{\noblock\rate}{\consize}\right\rceil+1$, the optimal outage diversity is achievable if for $\round=1, \ldots, \noround$, $\{\stairthre_\stair: R \geq \stairthre_\stair \geq \fbthre(\fbvec_{\round-1})\} \subseteq \{\fbthre(\fbvec_\round), \fbvec_\round \in \{1, \ldots, \nolevel-1\}^\round\}$.
  \item Furthermore, the outage probability in round $\round+1$ is dominated by the events with $\acinfo{\round} \in \setopenr{0, \frac{\consize}{\noblock}} \cup \setopenr{\stairthre_\stairrate, \rate}$, where $\stairrate=\left\lfloor\frac{\noblock\rate}{\consize}\right\rfloor$.  Therefore, for systems with $\nolevel > \left\lceil \frac{\noblock\rate}{\consize}\right\rceil+1$ feedback levels, the feedback thresholds should give higher priority to quantizing the aforementioned region to improve outage performance. 
  \end{itemize}
\end{remark}
  We now prove that the rate-diversity-delay tradeoff $\diver_{\round}(\rate)$ is achievable by using random codes, as given by the following theorem. 
\begin{theorem}
\label{the:achieve_Kbit}
Consider INR-ARQ transmission over the MIMO block-fading channel in (\ref{eq:channel_model}) using constellation $\conste$ of size $2^{\consize}$ and the transmission scheme described in Section \ref{sec:inr-arq-system} with power constraint $\txpow$ given in (\ref{eq:pow_constraint}).  Assume that the number of feedback levels is $\nolevel \geq \left\lceil\frac{\noblock\rate}{\consize}\right\rceil+1$.  With random-coding schemes and $\Blength \to \infty$, for large $\txpow$, the word error probability $\prober(\round)$ at round $\round$ satisfies
$\prober(\round) \doteq \txpow^{-\achievediver_{\round}(\rate)},$ 
where 
\begin{equation}
\label{eq:achieve_RDT}
\achievediver_{\round}(\rate)= (1+ \noblock\notx\norx)^{\round-1}\left(\rdiver(\rate)+1\right)-1 
\end{equation}
 is the achievable SNR-exponent and $\rdiver(\round)$ is given in Theorem \ref{the:asym_BF}. 
\end{theorem}
\begin{proof}
With a random coding scheme and $\Blength \to \infty$, the codeword is correctly decoded with probability one  at round $\round$ if $\acinfo{\round} > \rate$ \cite{CaireTuninetti2001,KnoppHumblet2000}, in which case, the receiver feeds back an ACK (in contrast to the outage case, where an ACK is fed back if $\acinfo{\round} \geq \rate$).  The proof then follows similar arguments as the proof of Theorem \ref{the:out_diver_K}, noting from Theorem \ref{the:asym_BF} that $\myPr{\info_{\round} \leq \info} \doteq \txpow_{\round}^{-\rdiver(\info)}$. 
\end{proof}

Theorem \ref{the:achieve_Kbit} shows that the rate-diversity-delay tradeoff $\diver_\round(\rate)$ stated in Theorem \ref{the:out_diver_K} is achievable  with random codes using  the transmission scheme described in Section \ref{sec:inr-arq-system} when $\frac{\noblock\rate}{\consize}$ is not an integer; and then, the optimal rate-diversity-delay tradeoff is given by (\ref{eq:achieve_RDT}).   Furthermore, the optimal outage diversity and SNR-exponent of INR-ARQ transmission with delay constraint $\noround$ is similarly characterized by $\diver_\noround(\rate)$ and $\achievediver_\noround(\rate)$ given in (\ref{eq:out_diver_K}) and (\ref{eq:achieve_RDT}), respectively. 
   

\subsection{One-bit MIMO ARQ}
In an  INR-ARQ system with one-bit ACK/NACK feedback (classical INR-ARQ), the optimal rate-diversity-delay tradeoff is given by the following. 
\begin{theorem}
\label{the:onebit}
Consider INR-ARQ transmission over the MIMO block-fading channel in (\ref{eq:channel_model}) using constellation $\conste$ of size $2^{\consize}$ and the transmission scheme described in Section \ref{sec:inr-arq-system}, where a codeword is considered successfully delivered at  round $\round$ if $\acinfo{\round} \geq \rate$.  Assume that  the number of feedback levels is $\nolevel=2$.   Subject to the power constraint in (\ref{eq:pow_constraint}), the optimal rate-diversity-delay tradeoff is given by
\begin{align}
\diveronebit_1(\rate)&= \tdiver(\rate)\\
\label{eq:diver_1bit}\diveronebit_{\round}(\rate)&= \noblock\notx\norx\left(\round-1 + \sum_{\roundrun=1}^{\round-2}\diveronebit_{\roundrun}(\rate)\right)+ (1+\diveronebit_{\round-1}(\rate))\diveronebit_1(\rate), ~~~ \round \geq 2. 
\end{align}
for all $\rate$ such that $\diveronebit_1(\rate)$ is continuous.  Furthermore, the rate-diversity-delay tradeoff $\diveronebit_{\round}(\rate)$ is achievable when $\frac{\noblock\rate}{\consize}$ is not an integer. 
\end{theorem}
\begin{proof}
The proof follows the same arguments as that of Theorems \ref{the:out_diver_K} and \ref{the:achieve_Kbit}, with only two feedback levels at 0 and $\rate$, respectively. 
\end{proof}
The outage diversity and optimal SNR-exponent of the INR-ARQ system with $\nolevel =2$ at rate $\rate$ is thus given by $\diveronebit_{\noround}(\rate)$ when $\frac{\noblock\rate}{\consize}$ is not an integer. 

%

%
\subsection{Numerical Results}
We numerically compare the optimal rate-diversity-delay tradeoff of INR-ARQ systems with $\nolevel \geq \left\lceil\frac{\noblock\rate}{\consize}\right\rceil+1$, and with $\nolevel=2$ as well as the optimal tradeoff of an INR-ARQ system with constant transmit power.  The optimal rate-diversity-delay tradeoff  $\diver_{\noround}(\rate)$ and $\diveronebit_{\noround}(\rate)$ for INR-ARQ transmission with $\noround=1, 2, 3$ over the MIMO block-fading channel with $\notx=\norx=\noblock=2$ are illustrated in Figure \ref{fig:long_term}.      

  For an INR-ARQ system with delay constraint $\noround$ and constant transmit power (short-term power constraint), the outage probability $\pout(\noround)$ is the same as that obtained by transmission with rate $\frac{\rate}{\noround}$ over a block-fading channel with $\noblock\noround$ fading blocks \cite{ChuangGuillenRasmussenCollings2008}.  From Theorem \ref{the:asym_BF}, the optimal outage diversity $\underline{\diver}_{\noround}(\rate)$ is given by\footnote{The rate-diversity-delay tradeoff of \cite{ChuangGuillenRasmussenCollings2008} is larger than that given in (\ref{eq:diver_uniform}) since it is obtained with rotations, which increase the constellation size, complexity and peak-to-average power ratio.}
\begin{equation}
\label{eq:diver_uniform}
\divershort_{\noround}(\rate)= \norx\left(1+\left\lfloor\noblock\noround\left(\notx-\frac{\rate}{\noround\consize}\right)\right\rfloor\right), 
\end{equation}
and is achievable by random codes for all rates $\rate$ such that $\divershort_{\noround}(\rate)$ is continuous.  The rate-diversity-delay tradeoff of the INR-ARQ system with constant transmit power is plotted in Figure \ref{fig:short_term}.   Figure \ref{fig:RDDT} shows  an order-of-magnitude improvement in outage diversity of INR-ARQ  when a long-term power constraint is allowed.   Furthermore,  significant gains in outage diversity are provided by multi-bit feedback, especially at transmission rates $\rate$ close to $\notx\consize$.  Since high $\rate$ is particularly relevant in ARQ systems, the result suggests that multi-bit feedback will give significant gains in practical implementations. 
\section{Power Adaptation and Feedback Design}
\label{sec:power-adapt-feedb}
The design of optimal feedback and transmission rules for an ARQ system with multi-bit feedback includes joint optimization of the overall set of quantization thresholds $\{\fbthre(\fbvec_{\noround-1}), \fbvec_{\noround-1} \in \{0, \ldots, \nolevel-2\}^{\noround-1}\}$ and the corresponding power adaptive rule $\txpow_{\round}(\fbvec_{\round-1})$.  The optimal feedback and power adaptation rule is obtained by minimizing
\begin{equation}
\sum_{\fbvec_{\noround-1} } \prfb(\fbvec_{\noround-1}) \pout(\noround|\fbvec_{\noround-1})
\end{equation}
subject to the power constraint in (\ref{eq:pow_constraint}). 
To the best of our knowledge, the optimization problem is not analytically tractable.  We therefore propose to partition the design problem into two steps. 
\begin{description}
\item[Step 1:] \hspace{1mm} At round $\round$, determine a set of feedback thresholds $\fbthre([\fbvec_{\round-1}, \fb])$ for every feedback vector $\fbvec_{\round-1} \in \{1, \ldots, \nolevel-2\}^{\round-1}$.
\item[Step 2:] \hspace{1mm} Given the set of feedback thresholds in Step 1, determine the corresponding transmit power rule, minimizing the outage probability.
\end{description}
The above procedure suboptimally partitions the joint optimization problem into two sequential problems.  Moreover, in the following, each individual problem is also suboptimally solved.  
Nevertheless, this design procedure leads to a practically implementable algorithm that achieves  the optimal diversity derived in the previous section. 
\subsection{Selecting the Set of Feedback Thresholds}
\label{sec:select-set-feedb}
From the observations in Remark \ref{rem:set_thres}, we propose the following choice of feedback thresholds.  Consider the feedback levels at round $\round$ for a given feedback vector $\fbvec_{\round-1}$. Let $\tau \triangleq \left\lfloor \frac{\noblock\rate}{\consize}\right\rfloor$, $\stairthre_{\stair}= \frac{\consize \stair}{\noblock}$ and $\stair' \triangleq \left\lfloor \frac{\noblock \fbthre(\fbvec_{\round-1})}{\consize}\right\rfloor$.  The feedback thresholds in round $\round$, given $\fbvec_{\round-1}$ is then determined as follows. 

\begin{enumerate}
\item Place a threshold at $\fbthre([\fbvec_{\round-1}, 0]) = \fbthre(\fbvec_{\round-1})$, and at $\fbthre([\fbvec_{\round-1}, \nolevel-1])= R$;
\item Place $\stairrate-\stair'$ thresholds at
$\stairthre_{\stair}, \stair= \stair'+1, \ldots, \tau$;
\item Place the remaining $\nolevel-2-\stairrate+\stair'$ thresholds sequentially within
\[ \hspace{-0.2 in} \setopen{\stairthre_{\stairrate}, \rate},
\setopen{\fbthre(\fbvec_{\round-1}), \stairthre_{\stair'+1}},
\setopen{\stairthre_{\stairrate-1}, \stairthre_{\stairrate}},
\setopen{\stairthre_{\stair'+1}, \stairthre_{\stair'+2}},
\ldots\] 
until no more thresholds are left to place, and such that the thresholds uniformly partition each region.
\end{enumerate}
The procedure for choosing the thresholds $\fbthre(\fbvec_{\round})$, given the feedback vector $\fbvec_{\round-1}$, is illustrated in Figure \ref{fig:threshold_design}.   More particularly, the feedback thresholds for INR-ARQ transmission over the block-fading channel with $\notx=\norx=1$, $\noblock=2$, $\nolevel = 4$, $\noround=2$, and $\rate= 3.5$ using 16-QAM constellations are illustrated in Figure \ref{fig:threshold}, where $\fbthre(\fbvec_{\round-1})= \fbthre([\fbvec_{\round-1},0])$, and the values of $\fbthre(\fbvec_2)$ are reported in Table \ref{tab:example_fbthre}. 

\subsection{Power Adaptation}
\label{sec:power-allocation}
The suboptimal power adaptation rule is obtained from the following  simplifications. 
\begin{itemize}
\item  We consider a power constraint more stringent than the constraint in  (\ref{eq:pow_constraint}), 
\begin{equation}
\label{eq:stringent_pow_constraint}
\sum_{\fbvec_{\round} \in \{0, \ldots, \nolevel-1\}^{\round}}\prfb(\fbvec_{\round}) \txpow_{\round+1}(\fbvec_{\round})\leq \frac{\txpow}{\noround}, 
\end{equation} 
for $\fbvec_{\round} \in \{0, \ldots, \nolevel-1\}^{\round}, \round=0, \ldots, \noround-1$, where $\prfb(\fbvec_0)=1$ by definition.
 \item When feedback $\fbvec_{\round-1}$ is received, we have that $\acinfo{\round-1} \geq \fbthre(\fbvec_{\round-1})$.  Then, the feedback probability is approximated from (\ref{eq:prfb}) by replacing  $\acinfo{\round-1}$ with $\fbthre(\fbvec_{\round-1})$; and the outage probability can be upper bounded as 
\begin{equation}
\upout(\round|\fbvec_{\round-1}) \triangleq \myPr{\info_{\round}+ \fbthre(\fbvec_{\round-1}) < \rate}, 
\end{equation}
where $\info_{\round}$ is given by (\ref{eq:info_1round}) with $\txpow_{\round}= \txpow_{\round}( \fbvec_{\round-1})$. 
\item To further simplify the problem, we consider minimizing  $\upout(\round), \round=1, \ldots, \noround$ sequentially. 
 \end{itemize}
Based on the simplifications, the corresponding power adaptation rule $\txpow_{\round}(\fbvec_{\round-1})$ is obtained by solving 
\begin{equation}
\label{eq:opt_powallocate}
\begin{cases}
{\rm Minimize}& \sum_{\fbvec_{\round-1}}\prfb(\fbvec_{\round-1}) \upout(\round|\fbvec_{\round-1})\\
{\rm Subject ~ to}& \sum_{\fbvec_{\round-1}}\prfb(\fbvec_{\round-1})\txpow_{\round}(\fbvec_{\round-1}) \leq \frac{\txpow}{\noround}.
\end{cases}
\end{equation}
 The optimization problem is separable, and thus can be solved via a branch-and-bound simplex  algorithm using piece-wise linear approximation \cite{BazaraaSheraliShetty}.  For single-input multiple-output (SIMO) channels, the probabilities $\prfb(\fbvec_{\round-1})$ and $\upout(\round|\fbvec_{\round-1})$ in (\ref{eq:opt_powallocate}) can be approximated numerically by shifting the outage probability bounds in \cite{NguyenGuillenRasmussen2007it} according to the gap between the bound and the corresponding simulation curve at high SNR.  For MIMO channels, solving (\ref{eq:opt_powallocate}) requires tabulating the probabilities $\prfb(\fbvec_{\round-1})$ and $\upout(\round|\fbvec_{\round-1})$, which can be obtained from Monte-Carlo simulations. 

\subsection{Numerical Results}
\label{sec:numerical-results}
First consider SISO ($\notx=\norx=1$) INR-ARQ transmission with $\noround =2$ at rate $\rate=3.5$ over the block-fading channel in (\ref{eq:channel_model}) with  $\noblock =2$ using $16$-QAM input constellations.  The outage performance of systems with $\nolevel = 2, 3, 8, 16$ is illustrated in Figure \ref{fig:SISO}.  We observe that the outage diversity achieved by constant transmit power and by power adaptation for $\nolevel=2$  is $3$ and $4$ as given in (\ref{eq:diver_uniform}) and (\ref{eq:diver_1bit}), respectively.  For $\nolevel \geq 3$, the outage diversity is 5 as predicted from (\ref{eq:out_diver_K}).  This leads to significant improvement in outage performance for power adaptive ARQ transmission with multi-bit feedback at high $\power$.   Particularly, 2 dB gain in power is observed at outage probability $10^{-6}$ for $\nolevel \geq 8$.  Note that at low $\power$, the outage performance of systems with $\nolevel=2$ is outperformed by system with constant transmit power due to the simplifying assumption (\ref{eq:stringent_pow_constraint}). 

The outage performance of MIMO INR-ARQ transmission over the block-fading channel in (\ref{eq:channel_model}) with $\notx=2, \norx=1, \noblock=1, \rate= 7.5$ using 16-QAM input constellations is illustrated in Figures \ref{fig:MIMO_sim} and \ref{fig:MIMO_bound}, where Figure \ref{fig:MIMO_sim} shows the simulation results, and Figure \ref{fig:MIMO_bound} presents the upper bound obtained from  (\ref{eq:opt_powallocate}).  The simulation results in Figure \ref{fig:MIMO_sim} have yet to show the correct outage diversity ($\diver_\noround(\rate)= 5$ for $\nolevel \geq 3$ and $\diver_\noround(\rate)= 4$ for $\nolevel=2$).  However, they follow the bounds from (\ref{eq:opt_powallocate}), which approach the optimal outage diversity at higher SNR as shown in Figure \ref{fig:MIMO_bound}. 
Figure \ref{fig:MIMO_sim} shows that systems with power allocation significantly outperform that with constant transmit power.  Moreover, allowing additional feedback levels ($\nolevel \geq 3$) provides further gains in outage diversity and thus significant gains in outage performance at high SNR. 

In both cases, the simulation results suggest that increasing $\nolevel$ beyond 8 does not substantially improve the outage performance; and thus, even for $\nolevel=3$, the suboptimal choice of feedback thresholds in Section \ref{sec:select-set-feedb} performs within 1dB of systems with large $\nolevel$ and optimal thresholds.   
\section{Conclusions}
\label{sec:conclusion}
We have studied the outage performance of MIMO block-fading channels with and without employing the INR-ARQ strategy.  An information-theoretic multi-bit feedback INR-ARQ scheme is proposed based on the accumulative mutual information, which potentially improves the performance of INR-ARQ transmission with minimal extra overhead requirement compared to classical INR-ARQ.  The study on power adaptation  has revealed large gains in outage diversity provided by multi-bit feedback in INR-ARQ systems with a long-term power constraint.  More generally, the multi-bit feedback INR-ARQ based on accumulated mutual information may prove useful in obtaining the fundamental limit of multi-bit feedback INR-ARQ systems.  Furthermore, since the proposed scheme is a generalization to that in \cite{VisotskySunTripathi2005} and \cite{SteinerShamai2008}, it promises further gain from the throughput performance obtained in \cite{VisotskySunTripathi2005,SteinerShamai2008}.

\appendices
\section{Proof of Theorem \ref{the:asym_BF}}
\label{sec:proof-theor-refpr}
We first assume a genie-aided receiver that perfectly eliminates the interference between the transmit antennas. This results in $\notx$ parallel SIMO block-fading channels, each with $\norx$ receive antennas.  Let $\infoga$ be the realized input-output mutual information of the genie-aided channel, then $\infoga \geq \info$. Furthermore, from the analysis in \cite{KnoppHumblet2000,GuillenCaire2006,NguyenGuillenRasmussen2007it}, we have that 
\begin{equation}
\myPr{\infoga < \rate} \doteq \power^{-\tdiver(\rate)}. 
\end{equation}
Therefore, 
\begin{equation}
\myPr{\info< \rate} \dot{\geq} \power^{-\tdiver(\rate)}. 
\end{equation}
The proof is thus completed by proving that 
\begin{equation}
\label{eq:proof_achieve}
\myPr{\info \leq \rate} \doteq \power^{-\rdiver(\rate)}. 
\end{equation}
Following the arguments  in \cite{KnoppHumblet2000,GuillenCaire2006,NguyenGuillenRasmussen2007it}, we have that 
\begin{equation}
    \myPr{\infoga \leq \rate} \doteq \power^{-\rdiver(\rate)}
\end{equation}
and therefore, 
\begin{equation}
\label{eq:lowbound_achieve}
\myPr{\info\leq \rate} \dot{\geq} \power^{-\rdiver(\rate)}. 
\end{equation}

We now prove that $\myPr{\info \leq \rate} \dot{\leq} \power^{-\rdiver(\rate)}$.  Considering transmission over the block-fading channel in (\ref{eq:channel_model}) with random codes of rate $\rate$, where the $\Blength\noblock\notx$ coded symbols in $\boldsymbol{x}$ are drawn uniformly random from the constellation $\conste$. Let $\rprober$ be the word error probability achieved by random coding. We have from the random-coding achievability and the strong converse theorem \cite{CoverThomas2006,Arimoto1973,MalkamakiLeib1999} that for a channel realization $\chmat$, 
\begin{equation}
\rprober(\chmat)= \begin{cases}
1 &{\rm if~} \info < \rate\\
0 &{\rm if~} \info > \rate
\end{cases}
\end{equation}
when $\Blength \to \infty$.  Therefore, the word error probability of random codes satisfies
\begin{equation}
\label{eq:achieve_eq_rand}
\rprober = \myPr{\info \leq \rate}. 
\end{equation}
We now prove that $\rprober \dot{\leq} \power_{\round}^{-\rdiver(\info)}$.  Consider encoding and transmitting a message $m$ as a random codeword $\txsig$. Assuming that the channel realization is $\chmat$, the pairwise error probability between $\txsig$ and $\txsig'$ is bounded by \cite{ViterbiOmura1979}
\begin{equation}
\label{eq:PEP}
\PEP\left(\txsig \rightarrow \txsig'|\chmat\right)  \leq \exp\left(-\frac{1}{4}\dist^2(\txsig, \txsig', \chmat)\right), 
\end{equation}
where, by letting $\bartxpow= \frac{\txpow_{\round}}{\notx}$,
\begin{equation}
\dist^2(\txsig, \txsig', \chmat)= \sum_{\block=1}^{\noblock}\sum_{\cuind=1}^{\Blength}\sum_{\rxind=1}^{\norx}\left|\sum_{\txind=1}^{\notx}\sqrt{\bartxpow}h_{\block,\txind,\rxind}(\txsig_{\block,\txind,\cuind}-\txsig'_{\block,\txind,\cuind})\right|^2. 
\end{equation}
Here, $h_{\block,\txind,\rxind}$ is the channel gain from transmit antenna $\txind$ to receive antenna $\rxind$ in block $\block$, and $X_{\block,\txind,\cuind}$ is the coded symbol transmitted by antenna $\txind$ at time instant $\cuind$ of block $\block$.  Let us write $h_{\block, \txind,\rxind}= |h_{\block, \txind,\rxind}|e^{i\theta_{\block, \txind,\rxind}}$, where $i=\sqrt{-1}$. Further define a matrix of normalized fading gains $\boldsymbol{\alpha} \in \mathbb{R}^{\noblock\times\notx\times\norx}$ where $\alpha_{\block, \txind,\rxind}\triangleq -\frac{\log(|h_{\block, \txind,\rxind}|^2)}{\log(\bartxpow)}$, then
\begin{equation}
\dist^2(\txsig, \txsig', \chmat)= \sum_{\block=1}^\noblock\sum_{\cuind=1}^{\Blength}\sum_{\rxind=1}^{\norx}\left|\sum_{\txind=1}^{\notx}\bartxpow^{\frac{1-\alpha_{\block,\txind,\rxind}}{2}}e^{i\theta_{\block,\txind,\rxind}}(\txsig_{\block, \txind,\cuind}-\txsig'_{\block, \txind,\cuind})\right|^2. 
\end{equation}
By averaging (\ref{eq:PEP}) over the random coding ensemble, the pairwise error probability of random codes is
\begin{align}
\PEP^{(r)}(\txsig\rightarrow \txsig'|\chmat) &\leq \prod_{\block=1}^{\noblock}\left\{\frac{1}{2^{2\consize\notx}} \sum_{\mathbf{x} \in \conste^{\notx}}\sum_{\mathbf{x'}\in \conste^{\notx}}\exp\left(-\frac{1}{4}\sum_{\rxind=1}^{\norx}\left|\sum_{\txind=1}^{\notx}\bartxpow^{\frac{1-\alpha_{\block,\txind,\rxind}}{2}}e^{i\theta_{\block,\txind,\rxind}}(x_t-x_t')\right|^2\right)\right\}^{\Blength}\\
&\leq \exp\left(\noblock\consize\Blength\log(2)\left(-2\notx+\frac{1}{\noblock\consize}T(\bartxpow, \boldsymbol{\alpha})\right)\right), 
\end{align}
where $x_t$ is the $t^{\rm th}$ entry of vector $\mathbf{x}$ and
\begin{equation}
\label{eq:T_def}
T(\bartxpow, \boldsymbol{\alpha})\triangleq \sum_{\block=1}^\noblock\log_2 \left(\sum_{\mathbf{x} \in \conste^{\notx}}\sum_{\mathbf{x}' \in \conste^{\notx}} \exp\left(-\frac{1}{4}\sum_{\rxind=1}^{\norx}\left|\sum_{\txind=1}^{\notx} \bartxpow^{\frac{1-\alpha_{\block,\txind,\rxind}}{2}}e^{i\theta_{\block,\txind,\rxind}}(x_t-x'_t)\right|^2\right)\right).
\end{equation}
By summing over the $2^{\noblock\rate\Blength}-1$ possible error events, the union bound on the word error probability is given by 
\begin{equation}
\label{eq:union_Pe}
\rprober(\chmat)\leq \min\left\{1, \exp\left(\noblock\consize\Blength\log(2)\left(-2\notx + \frac{\rate}{\consize}+ \frac{1}{\noblock \consize}T(\bartxpow, \boldsymbol{\alpha})\right)\right)\right\}.
\end{equation}

For any $\epsilon > 0$, denote $\goodset_\block \triangleq \bigcup_{\rxind=1}^{\norx}\goodset_{\block,\rxind}$, and $\kappa_\block \triangleq|\goodset_\block|$, where
\begin{equation}
\goodset_{\block, \rxind}\triangleq \{\txind:\alpha_{\block, \txind,\rxind} \leq 1-\epsilon, \txind= 1, \ldots,\notx\}.
\end{equation}
Then, for any given $\rxind \in \{1, \ldots, \norx\}$, and letting $\alpha_{\block,\rxind}= \max\{\alpha_{\block,\txind,\rxind}, \txind \in \goodset_{\block,\rxind}\}$, we can write
\begin{align}
\lim_{\bartxpow\to \infty}\sum_{\txind=1}^{\notx}\bartxpow^{\frac{1-\alpha_{\block, \txind,\rxind}}{2}}e^{i\theta_{\block, \txind,\rxind}}(x_{\txind}-x_{\txind}') &\geq \lim_{\bartxpow \to \infty}\sum_{\substack{\txind \in \goodset_{\block,\rxind}\\ x_{\txind} \neq x_{\txind}'}} \bartxpow^{\frac{1-\alpha_{\block,\txind,\rxind}}{2}}e^{i\theta_{\block, \txind,\rxind}}(x_{\txind}-x_{\txind}')\\
&\geq \lim_{\bartxpow \to \infty} \bartxpow^{\frac{1-\alpha_{\block,\rxind}}{2}}\sum_{\substack{\txind \in \goodset_{\block,\rxind}\\ x_{\txind} \neq x_{\txind}'}} e^{i\theta_{\block,\txind,\rxind}}(x_{\txind}-x_{\txind}').
\end{align}
Since  the $\theta_{\block,\txind,\rxind}$'s are uniformly drawn from $[-\pi,\pi]$, we have that 
\begin{equation}
\lim_{\bartxpow\to \infty}\sum_{\txind=1}^{\notx}\bartxpow^{\frac{1-\alpha_{\block, \txind,\rxind}}{2}}e^{i\theta_{\block,\txind,\rxind}}(x_{\txind}-x_{\txind}') = \infty
\end{equation} 
with probability 1 if there exists $\txind \in \goodset_{\block,\rxind}$ such that $x_\txind \neq x_\txind'$.  Noting that $\kappa_\block= |\goodset_b|$, it follows from (\ref{eq:T_def}) that
\begin{align}
\lim_{\bartxpow \to \infty} T(\bartxpow, \boldsymbol{\alpha})&= \sum_{\block=1}^{\noblock}\lim_{\bartxpow\to\infty} \log_2\left(\sum_{\mathbf{x}\in \conste^{\notx}}\sum_{\substack{\mathbf{x'}\in \conste^\notx\\x_{\txind}'= x_\txind, \forall t \in \goodset_\block}} \exp\left(-\frac{1}{4}\sum_{\rxind=1}^{\norx}\left|\sum_{\txind=1}^{\notx}\bartxpow^{\frac{1-\alpha_{\block,\txind,\rxind}}{2}}e^{i\theta_{\block,\txind,\rxind}}(x_\txind-x_\txind')\right|^2\right)\right)\notag\\
&\leq \sum_{\block=1}^{\noblock}\log_2\left(2^{\consize\notx}2^{\consize(\notx-\kappa_\block)}\right)\notag\\
&=\sum_{\block=1}^\noblock\consize(2\notx-\kappa_\block). 
\end{align}
Thus, the error probability in (\ref{eq:union_Pe}) is asymptotically upper-bounded by
\begin{equation}
\lim_{\bartxpow \to \infty}\rprober(\chmat) \leq \min\left\{1, \exp\left(-\noblock\consize\Blength\log(2)\left(\frac{1}{\noblock}\sum_{\block=1}^\noblock \kappa_\block-\frac{\rate}{\consize}\right)\right)\right\}.
\end{equation}
Let $\mathcal{B}^{(\epsilon)} \triangleq \left\{\boldsymbol{\alpha} \in \mathbb{R}^{\noblock\times\notx\times\norx}:\sum_{\block=1}^\noblock \kappa_\block \leq \frac{\noblock\rate}{\consize}\right\}$ be the outage set.  By averaging over the fading matrix and letting $\Blength \to \infty$, the error probability is bounded by
\begin{equation}
\rprober\leq \int_{\boldsymbol{\alpha} \in \mathcal{B}^{(\epsilon)}} f_{\boldsymbol{\alpha}}(\boldsymbol{\alpha})d\boldsymbol{\alpha}, 
\end{equation}
where $f_{\boldsymbol{\alpha}}(\boldsymbol{\alpha})$ is the joint pdf of the random vector $\boldsymbol{\alpha}$.  Following the analysis in \cite{GuillenCaire2006}, and letting $\Blength \to \infty$, the SNR-exponent for the case of using random codes is lower bounded by 
\begin{align}
\inf_{\boldsymbol{\alpha}\in \mathcal{B}^{(\epsilon)} \cap \mathbb{R}_+^{\noblock\norx \times \noblock\notx}}\left\{\sum_{\block=1}^{\noblock}\sum_{\txind=1}^{\notx}\sum_{\rxind=1}^{\norx} \alpha_{\block,\txind,\rxind}\right\} &= \norx\left(\noblock\notx-\left\lfloor \frac{\noblock\rate}{\consize}\right\rfloor\right)(1-\epsilon)\\&
 = \norx\left\lceil\noblock\left(\notx-\frac{\rate}{\consize}\right)\right\rceil (1-\epsilon).
\end{align}
Thus, by letting $\epsilon \downarrow 0$, the outage diversity $\rdiver(\rate)$ is achievable using random codes.  Therefore we have from (\ref{eq:achieve_eq_rand}) that 
\begin{equation}
\myPr{\info \leq \rate} \dot{\leq} \bartxpow^{-\rdiver(\rate)} \doteq \txpow^{-\rdiver(\rate)}.
\end{equation}
Thus, (\ref{eq:proof_achieve}) is obtained from (\ref{eq:lowbound_achieve}). 
\section{Proof of Theorem \ref{the:out_diver_K}}
\label{sec:proof-out-diver-K}

A sketch of the proof is given as follows.  We first lower-bound the outage diversity by considering a suboptimal ARQ system with $\underline{K} = \left\lceil\frac{\noblock \rate}{\consize}\right\rceil+1$ feedback levels, where the quantization thresholds are placed at $\fbthre([\fbvec_{\round-1}, \fb_{\round}]) = \frac{\fb_{\round} \consize}{\noblock}, \fb_{\round}= 0, \ldots, \left\lfloor\frac{B \rate}{\consize}\right\rfloor$.  Using Theorem \ref{the:asym_BF}, we prove by induction that the outage diversity of the suboptimal ARQ system at round $\round$ is $\diver_\round(\rate)$.  

Conversely, consider an optimal INR-ARQ system with $\nolevel \geq \left\lceil \frac{\noblock\rate}{\consize}\right\rceil+1$ feedback levels.  The outage performance of the system can be improved by adding  $\left\lfloor \frac{\noblock \rate}{\consize}\right\rfloor+1$  extra quantization thresholds (and corresponding feedback indices) at $\frac{\stair \consize}{\noblock}, \stair=0, \ldots, \left\lfloor\frac{\noblock\rate}{\consize}\right\rfloor$.  Using Theorem \ref{the:asym_BF}, we prove by induction that the outage diversity at round $\round$ of the improved systems (with $\nolevel+ \left\lfloor \frac{\noblock\rate}{\consize}\right\rfloor+1 $ feedback levels) is also given by $\diver_{\round}(\rate)$.  Therefore, $\diver_{\round}(\rate)$ is the optimal outage diversity at round $\round$ for an ARQ system with $\nolevel \geq \left\lceil\frac{\noblock \rate}{\consize}\right\rceil+1$  feedback levels. 
\subsection{Lower bound on the optimal outage diversity}
\label{sec:lower-bound-outage}
To get a lower bound to the outage diversity, consider an ARQ system with $K=\left\lceil\frac{\noblock\rate}{\consize}\right\rceil+1$ feedback levels, where the following (suboptimal) set of feedback thresholds is employed, 
\begin{equation}
\fbthre(\fbvec_{\round}) =\begin{cases} 
\stairthre_{\fb_{\round}}, &0 \leq \fb_{\round}< \nolevel-1\\
\rate, & \fb_{\round}= \nolevel-1, 
\end{cases}
\end{equation}
with $\stairthre_{\stair}= \frac{\stair \consize}{\noblock}$.  In this case, feedback index $\fb_{\round} =\stair$ is delivered at round $\round$ if $\acinfo{\round} \in \setopenr{\stairthre_{\stair}, \stairthre_{\stair+1}}$, regardless of the realized feedback indices of the previous rounds.  At round $\round$, the transmit power is  suboptimally adapted to the feedback index $\fb_{\round-1}$ as $\txpow_{\round} = \txpow_{\round}(\fb_{\round-1})$, where
\begin{equation}
\label{eq:sub_power_K}
\txpow_{\round}(\fb_{\round-1}) = \begin{cases}
\frac{\txpow}{\nolevel\noround\myPr{\acinfo{\round-1}\in \setopenr{\stairthre_{\fb_{\round-1}}, \stairthre_{\fb_{\round-1}+1}}}}, & \fb_{\round-1} < K-1\\
0, &{\rm otherwise}. 
\end{cases}
\end{equation}
The power adaptation rule in (\ref{eq:sub_power_K}) satisfies the power constraint in (\ref{eq:pow_constraint}).  We now derive the outage diversity achieved by the aforementioned system. 

For $\info \in \setopen{\stairthre_{\stair}, \stairthre_{\stair+1}}$, we have from Theorem \ref{the:asym_BF} that 
\begin{equation}
\label{eq:K_diver_round1}
\myPr{\info_1 < \info} \doteq \myPr{\info_{1} \in \setopenr{\stairthre_{\stair}, \stairthre_{\stair+1}}} \doteq \txpow^{-\Kdiverint_1(\stair)}, 
\end{equation}
where $\Kdiverint_{1}(\stair)\triangleq \diver(\stairthre_{\stair+1})= \norx(\noblock \notx -\stair)$. 

For $\stair= 0, \ldots, \noblock\notx-1$ and a given $\info \in \setopen{\stairthre_{\stair}, \stairthre_{\stair+1}}$, we now prove by induction that for $\round=1, \ldots, \noround$, 
\begin{equation}
\label{eq:K_diver_asump}
\myPr{\acinfo{\round} < \info} \doteq \myPr{\acinfo{\round} \in \setopenr{\stairthre_{\stair}, \stairthre_{\stair+1}}} \doteq \txpow^{-\Kdiverint_{\round}(\stair)}, 
\end{equation}
where $\Kdiverint_{\round}(\stair) = \diver_\round\left(\stairthre_{\stair+1}\right)$ is given in (\ref{eq:out_diver_K}). 

Equation (\ref{eq:K_diver_round1}) shows that (\ref{eq:K_diver_asump}) is correct at round 1.  Assume now that (\ref{eq:K_diver_asump}) is correct at round $\round$. From (\ref{eq:sub_power_K}) we have that 
\begin{equation}
\label{eq:txpow_asump}
\txpow_{\round+1}(\stair) = \frac{\txpow}{\nolevel\noround\myPr{\acinfo{\round} \in \setopenr{\stairthre_{\stair}, \stairthre_{\stair+1}}}} \doteq \txpow^{1+\Kdiverint_{\round}(\stair)}. 
\end{equation}
Therefore, for $\info \in \setopen{\stairthre_{\stair}, \stairthre_{\stair+1}}$, \begin{align}
&\myPr{\acinfo{\round+1} < \info} = \sum_{j=0}^{\stair}   \myPr{\acinfo{\round} \in \setopenr{\stairthre_j,   \stairthre_j+\info-\stairthre_{\stair}}}   \myPr{\info_{\round+1} < \info-\acinfo{\round}   \Big|\acinfo{\round} \in \setopenr{\stairthre_j,   \stairthre_j+\info-\stairthre_{\stair}}} + \notag\\ 
\label{eq:bound1}&\hspace{0.2 in}\sum_{j=0}^{\stair} \myPr{\acinfo{\round}\in \setopenr{\stairthre_j+\info-\stairthre_{\stair}, \stairthre_{j+1}}}\myPr{\info_{\round+1} < \info-\acinfo{\round}\Big|\acinfo{\round} \in \setopenr{\stairthre_j+\info-\stairthre_{\stair},\stairthre_{j+1}}}.
\end{align}
Given $\acinfo{\round} \in \setopenr{\stairthre_j, \stairthre_j+\info-\stairthre_{\stair}}$ and $\info \in \setopen{\stairthre_{\stair}, \stairthre_{\stair+1}}$, we have that $\info-\acinfo{\round} \in \setopen{\stairthre_{\stair-j}, \stairthre_{\stair-j+1}}$.  Therefore,  by applying Theorem \ref{the:asym_BF}, and noting the transmit power in (\ref{eq:txpow_asump}), we have that
\begin{equation}
\label{eq:upbound2}
\myPr{\info_{\round+1} < \info-\acinfo{\round}\Big|\acinfo{\round}
\in \setopenr{\stairthre_j, \stairthre_j+\info-\stairthre_{\stair}}}
\doteq \txpow^{-(1+\Kdiverint_{\round}(j))\norx(\noblock\notx-\stair+j)}.
\end{equation}
Since (\ref{eq:K_diver_asump}) is assumed at round $\round$, the first summation dominates in (\ref{eq:bound1}).  Thus from (\ref{eq:upbound2}), we have that 
\begin{equation}
\label{eq:asymp_exp_outage} \myPr{\acinfo{\round+1}<\info} \doteq \sum_{j=0}^t \txpow^{-\Kdiverint_{\round}(j)-\left[1+\Kdiverint_{\round}(j)\right]\norx(\noblock\notx-\stair+j)}. 
\end{equation}
The asymptotic exponent in (\ref{eq:asymp_exp_outage}) is given by
 \begin{align}
 &\hspace{-0.5 in}\min_{j=0, \ldots,  \stair}\Kdiverint_{\round}(j)+\left[1+\Kdiverint_{\round}(j)\right]\norx(\noblock\notx-\stair+j)\\ 
\label{eq:plug_assump}& = \min_{j=0, \ldots, \stair} -1+ (1+\noblock\norx\notx)^{\round-1}\left[1+\norx(\noblock\notx-j)\right]\left[1+\norx(\noblock\notx-\stair+j)\right]\\
\label{eq:takemin}&=-1+ (1+\noblock\notx\norx)^{\round}\left[1+\norx(\noblock\notx-\stair)\right]\\
&=\Kdiverint_{\round+1}(\stair),
\end{align}
where (\ref{eq:plug_assump}) follows from assumption $\Kdiverint_\round(j)= \diver_{\round}(\stairthre_{j+1})$ in (\ref{eq:K_diver_asump}), and (\ref{eq:takemin}) follows since the minimum in (\ref{eq:plug_assump}) is achieved with either $j=0$ or $j=t$.  Therefore, from (\ref{eq:asymp_exp_outage}), 
\begin{equation}
  \label{eq:first_con_induct_proof}
  \myPr{\acinfo{\round+1}< \info} \doteq \power^{-\Kdiverint_{\round+1}(\stair)}, 
\end{equation}
where $\Kdiverint_{\round+1}(\stair) = \diver_{\round+1}(\stairthre_{\stair+1})$ in (\ref{eq:out_diver_K}).  Thus, (\ref{eq:K_diver_asump}) is correct for $\round=1, \ldots, \noround$ by induction.  Consequently, for any $\rate \in \setopen{\stairthre_{\stairrate},\stairthre_{\stairrate+1}}$, we have that 
\begin{equation}
\myPr{\acinfo{\round} < \rate} \doteq \txpow^{-\Kdiverint_{\round}(\stairrate)} = \power^{-\diver_{\round}(\stairthre_{\stairrate+1})},
\end{equation}
and thus, the diversity in (\ref{eq:out_diver_K}) is achieved by the ARQ system with $\stairrate+2 = \left\lceil \frac{\noblock\rate}{\consize}\right\rceil+1$ feedback levels.  

Noting when $\myPr{\acinfo{\round+1} < \rate} \doteq \power^{-\Kdiverint_\round(\stairrate)}$, the outage probability at round $\round$ is dominated by the events with $j=0$ and $j=\stairrate$ in (\ref{eq:asymp_exp_outage}), which correspond to the events with $\acinfo{\round}\in \setopenr{0, \stairthre_1} \cup \setopenr{\stairthre_\stairrate, \rate}$.  The observation is useful for designing the feedback thresholds for the system, as summarized in Remark \ref{rem:set_thres}. 
\subsection{Upper bound on the optimal outage diversity}
Conversely, we derive an upper bound to the outage diversity achieved by a system with optimal feedback threshold $\fbthre(\fbvec_{\round})$ with $K$ levels per transmission round.  We first assume that $\rate \in \setopen{\stairthre_\stairrate,\stairthre_{\stairrate+1}}$ for some $\stairrate \in \{0, \ldots, \noblock\notx-1\}$. Consider improving the performance of the system by employing a feedback threshold set $\ifbthre{\fbvec_{\round}}$ with $\inolevel= K+\stairrate+1$ feedback levels per ARQ round by adding $\stairrate+1$ levels to the optimal feedback threshold set $\{\fbthre(\fbvec_{\round})\}$. The extra $\stairrate+1$ levels are located at $\stairthre_{\stair}= \frac{\stair M}{\noblock}, \stair=0, \ldots, \stairrate$.




Let $\fbthresetr{\round-1}(k) \triangleq \setopenr{\ifbthre{[\fbvec_{\round-1}, k]}, \ifbthre{[\fbvec_{\round-1}, k+1]}}, \round = 1, \ldots, \noround, k=0, \ldots, \inolevel-2$, and further let $\fbthresetl{\round-1}(k) \triangleq \setopen{\ifbthre{[\fbvec_{\round-1}, k]}, \ifbthre{[\fbvec_{\round-1}, k+1]}}$.   Then, given that the  feedback  vector at round $\round-1$ is $\fbvec_{\round-1}$, the receiver delivers feedback index $\inolevel-1$ if $\acinfo{\round} \geq \ifbthre{[\fbvec_{\round-1}, \inolevel-1]}= \rate$; otherwise, it delivers index $\fb_{\round}$, where $\fb_{\round}$ is chosen such that $\acinfo{\round} \in \fbthresetr{\round-1}(\fb_{\round})$.  

From the power constraint (\ref{eq:pow_constraint}), the optimal power allocation rule is upper-bounded by 
\begin{equation}
\label{eq:asymp_pow}
  \overline{\txpow}_\round(\fbvec_{\round-1}) =
  \begin{cases}
    \power, &\round=1\\    
    \frac{\power}{\myPr{\acinfo{\round-1} \in \fbthresetr{\round-2}(\fb_{\round-1})}}, &\fb_{\round-1} < \nolevel-1\\
    0, &{\rm otherwise.}
  \end{cases}
\end{equation}
Meanwhile, the power adaptation rule 
\begin{equation}
  \underline{\power}_{\round}(\fbvec_{\round-1}) =
  \begin{cases}
    \frac{\power}{\noround}, &\round=1\\    
    \frac{\power}{\inolevel\noround\myPr{\acinfo{\round-1} \in \fbthresetr{\round-2}(\fb_{\round-1})}}, &\fb_{\round-1} < \nolevel-1\\
    0, &{\rm otherwise.}
  \end{cases}
\end{equation}
satisfies the power constraint in (\ref{eq:pow_constraint}).  Therefore, the optimal power allocation rule asymptotically satisfies $ \txpow_{\round}(\fbvec_{\round-1}) \doteq \overline{\txpow}_{\round}(\fbvec_{\round-1})$ given in (\ref{eq:asymp_pow}). 


For $\stair=0, \ldots, \stairrate$, let $\indexset{\round-1}(\stair)= \left\{k \in \{1, \ldots, \inolevel-2\}: \fbthresetr{\round-1}(k) \subseteq \setopenr{\stairthre_{\stair}, \stairthre_{\stair+1}}\right\}$. Since $\stairthre_{\stair},$ for $\stair= 1, \ldots, \stairrate$, belongs to the set of thresholds $\{\ifbthre{[\fbvec_{\round-1}, \fb_{\round}]}, \fb_{\round}= 0, \ldots, \inolevel-1\}$,  $\fbthresetl{\round-1}(k) \subseteq \setopen{\stairthre_{\stair}, \stairthre_{\stair} +1}$ for some $\stair  \in \{1, \ldots, \stairrate\}$.  
Applying Theorem \ref{the:asym_BF}, 
for any $\info \in \setopen{\stairthre_{\stair}, \stairthre_{\stair+1}}$ and $k \in \indexset{0}(\stair)$, we have that
\begin{align}
\myPr{\info_1 < \info}  & \doteq \txpow^{-\norx(\noblock\notx-\stair)} \doteq \txpow^{-\Kdiverint_1(\stair)}\\
\myPr{\info_1 \in \fbthresetr{0}(k)}&\doteq \myPr{\info < \ifbthre{[k+1]}} \doteq \txpow^{-\Kdiverint_1(\stair)},
\end{align}
where $\Kdiverint_1(\stair)= \diver_1(\stairthre_{\stair+1})$ given in (\ref{eq:out_diver_K}).

 For the induction proof, assume that when $\info \in \setopen{\stairthre_\stair, \stairthre_{\stair+1}}$ and $\fb \in \indexset{\round-1}(t)$, we have
 \begin{equation}
 \label{eq:induct_assump}
 \myPr{\acinfo{\round} < \info} \doteq \myPr{\acinfo{\round}
 \in \fbthresetr{\round-1}(k)} \doteq \txpow^{-\Kdiverint_{\round}(\stair)}, 
 \end{equation}
where $\Kdiverint_\round(\stair)= \diver_\round(\stairthre_{\stair+1})$ given in (\ref{eq:out_diver_K}). 
 The assumption is correct for $\round=1$. We prove that (\ref{eq:induct_assump}) is also valid at round $\round+1$.  In fact, considering $\info \in \setopen{\stairthre_\stair,\stairthre_{\stair+1}}$, we have
\begin{align}
&\myPr{\acinfo{\round+1} < \info} = \sum_{j=0}^{\stair} \myPr{\acinfo{\round} \in \setopenr{\stairthre_{j}, \stairthre_{j}+\info-\stairthre_{\stair}}} \myPr{\info_{\round+1} < \info-\acinfo{\round}\Big|\acinfo{\round} \in \setopenr{\stairthre_{j}, \stairthre_{j}+\info-\stairthre_{\stair}}}\notag\\
&\hspace{0.7 in}+\sum_{j=0}^{\stair} \myPr{\acinfo{\round} \in \setopenr{\stairthre_{j}+\info-\stairthre_{\stair}, \stairthre_{j+1}}} \myPr{\info_{\round+1} < \info-\acinfo{\round}\Big|\acinfo{\round} \in \setopenr{\stairthre_{j}+\info-\stairthre_{\stair}, \stairthre_{j+1}}}. \notag
\end{align}

From assumption (\ref{eq:induct_assump}) and power allocation rule (\ref{eq:asymp_pow}),  when $\acinfo{\round} \in \fbthresetr{\round-1}(\fb_{\round})$, the transmit power in round $\round+1$ is $\txpow_{\round+1}\doteq\frac{\txpow}{\myPr{\acinfo{\round} \in \fbthresetr{\round-1}(\fb_{\round})}} \doteq \txpow^{1+\Kdiverint_{\round}(j)}$  for all $\fb_{\round} \in \indexset{\round}(j)$. Therefore, when $\acinfo{\round} \in \setopenr{\stairthre_{j}, \stairthre_{j+1}}$, $\txpow_{\round+1} \doteq \txpow^{1+\Kdiverint_{\round}(j)}$.  Thus, with similar arguments that are used to derive (\ref{eq:upbound2}), we have that 
 \begin{equation}
  \myPr{\acinfo{\round+1}<\info} \doteq \sum_{j=0}^t \txpow^{-(\Kdiverint_{\round}(j)+(1+\Kdiverint_{\round}(j))\norx(\noblock\notx-\stair+j))} 
 \end{equation}
 as given in (\ref{eq:asymp_exp_outage}). Therefore, following the steps used to derive (\ref{eq:first_con_induct_proof}), we have that 
 \begin{equation}
   \label{eq:ind_proof_thre}
   \myPr{\acinfo{\round+1} < \info} \doteq \txpow^{-\Kdiverint_{\round+1}(\stair)}
 \end{equation}
for $\info \in \setopen{\stairthre_{\stair}, \stairthre_{\stair+1}}$. It follows that
\begin{equation}
\label{eq:ind_proof_set}
\myPr{\acinfo{\round+1} \in \fbthresetr{\round}(k)} \doteq \myPr{\acinfo{\round+1} < \ifbthre{[\fbvec_{\round}, k]}} \doteq \txpow^{-\Kdiverint_{\round+1}(\stair)}
\end{equation}
 for all $k \in \indexset{\round}(\stair)$.  The results in (\ref{eq:ind_proof_thre}) and (\ref{eq:ind_proof_set}) prove that assumption (\ref{eq:induct_assump}) is valid at round $\round+1$, and thus by mathematical induction, (\ref{eq:induct_assump}) is valid for $\round=1, \ldots, \noround$. 

 Since $\rate \in \setopen{\stairthre_{\stairrate}, \stairthre_{\stairrate+1}}$,    
 \begin{equation}
\myPr{\acinfo{\round} < \rate} \doteq \txpow^{-\Kdiverint_{\round}(\stairrate)} \doteq \txpow^{-\diver_\round(\stairthre_{\stairrate+1})}, 
 \end{equation}
which proves that the outage diversity of the system with $\inolevel-$level feedback is the same as that given in (\ref{eq:out_diver_K}). 

\bibliographystyle{IEEEtran}
\bibliography{bib_database}

\newpage

\begin{table}
\caption{Feedback thresholds for $\notx=\norx=1, \noblock =2, \noround=2, \rate = 3.5$.}
\label{tab:example_fbthre}
\begin{center}
\begin{tabular}{|c|c|c|c|}
\hline
&$\fb_2=0$&$\fb_2=1$&$\fb_2=2$\\
\hline
$\fb_1=0$&0&2&2.75\\
$\fb_1=1$&2&2.5&3.0\\
$\fb_1=2$&2.75&3.0&3.25\\
\hline
\end{tabular}
\end{center}
\end{table}

\newpage
\begin{figure}[tbf]
\centering
\input{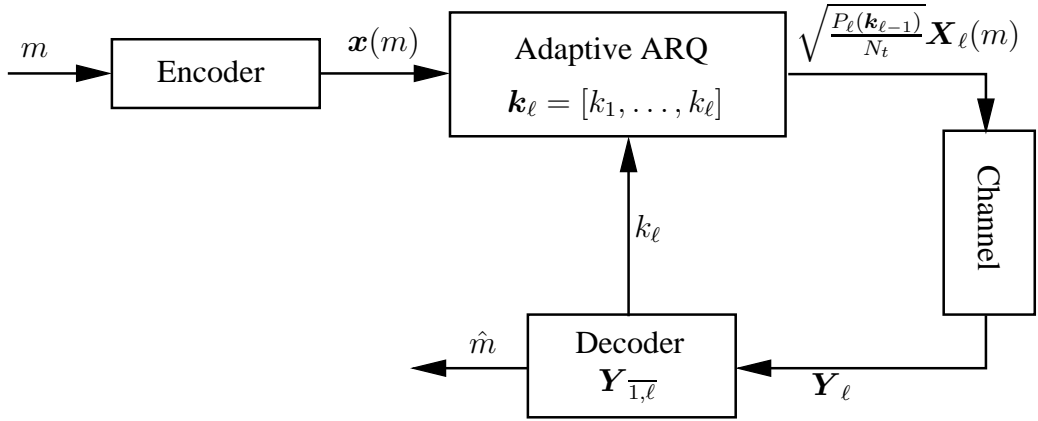}
\caption{The  INR-ARQ system with multi-bit feedback.}
\label{fig:ARQ_sys}
\end{figure}

\begin{figure}[tbf]
\centering
\input{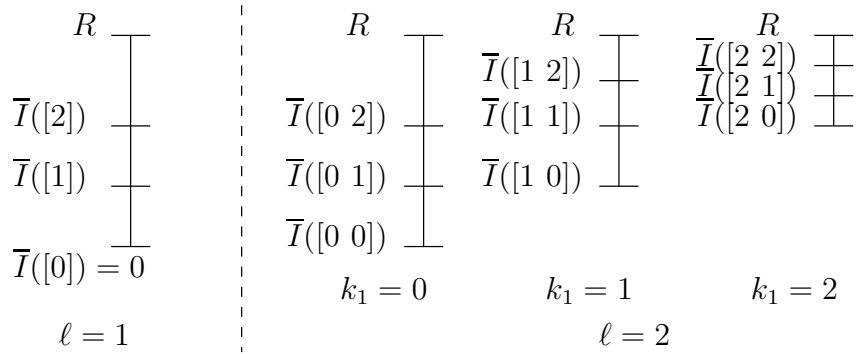}
\caption{An example of feedback thresholds.}
\label{fig:threshold}
\end{figure}

\begin{figure*}
\begin{center}
\subfigure[Long-term power constraint tradeoff.]{
\includegraphics[width=0.75\columnwidth]{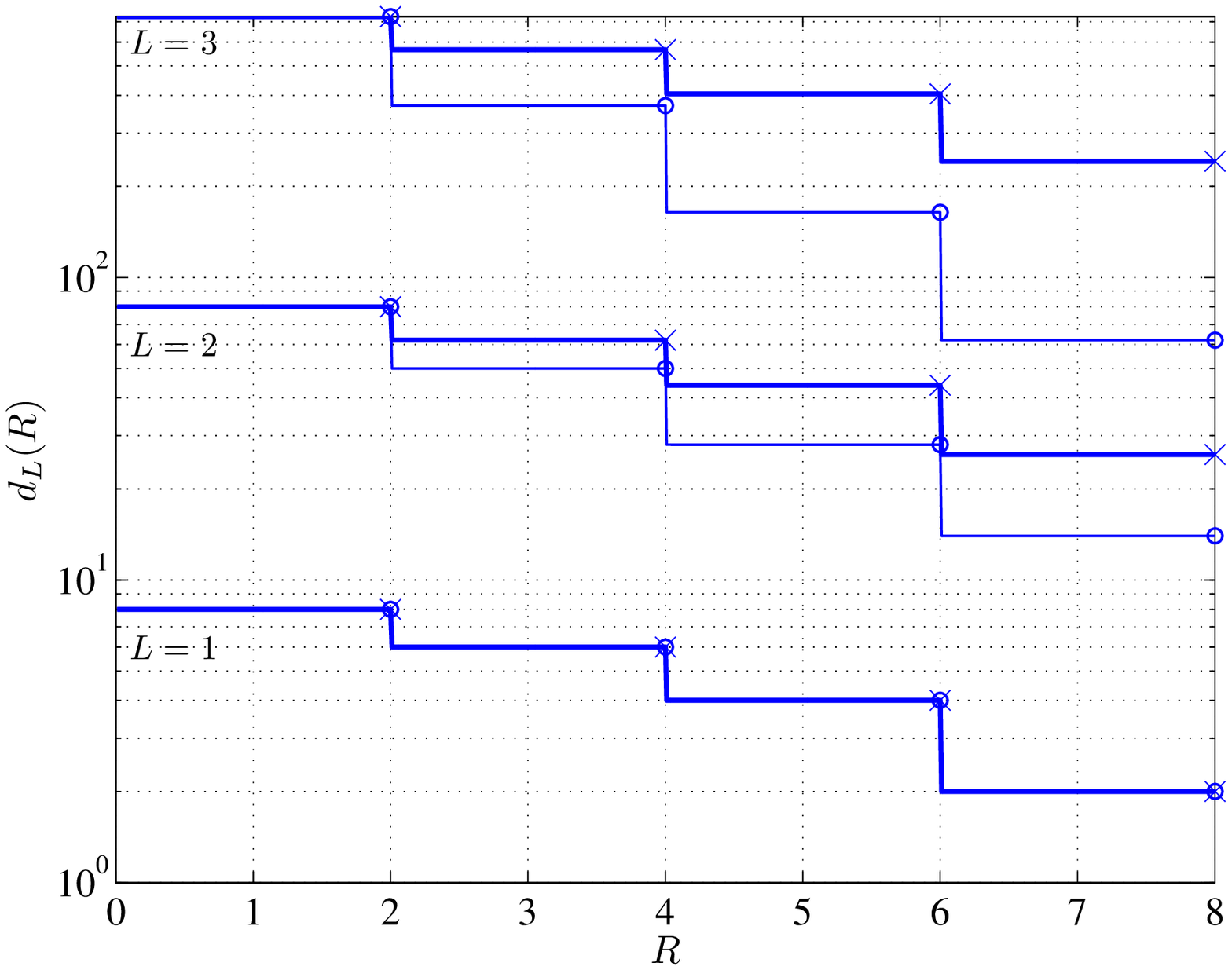}\label{fig:long_term}}
\subfigure[Constant transmit power tradeoff.]{
\includegraphics[width=0.75\columnwidth]{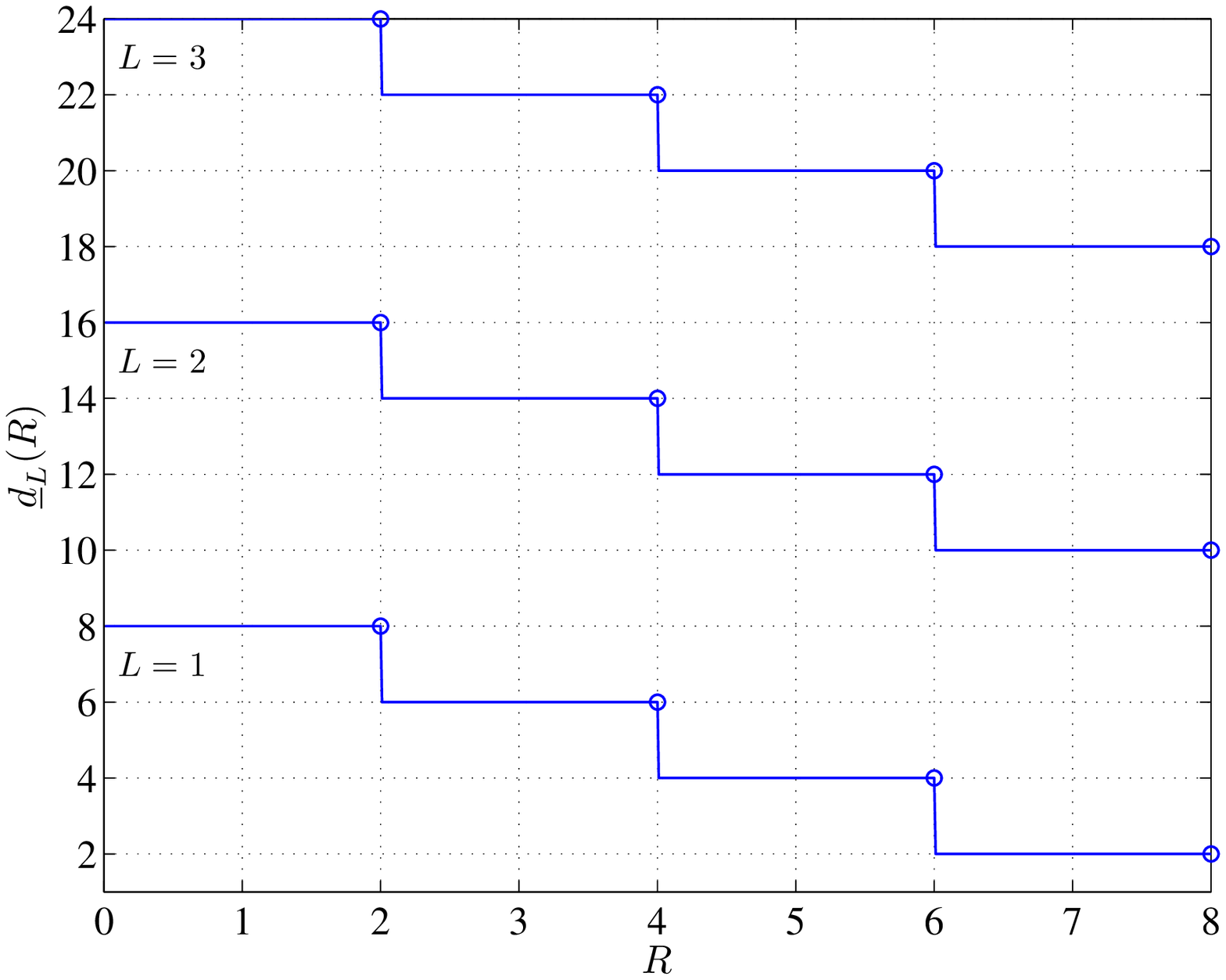}\label{fig:short_term}}
\caption{Optimal rate-diversity-delay tradeoff of ARQ transmission with long-term power constraint  (a) and constant power (b).  16-QAM is used over a  MIMO block-fading channel with $\notx=\norx=2, \noblock=2, \noround=1, 2, 3$.  Thick and thin lines in (a) represent the optimal tradeoffs $\diver_{\noround}(\rate)$ achieved by multi-bit feedback ($\nolevel \geq \left\lceil\noblock\rate/\consize\right\rceil+1$) and $\diveronebit_{\noround}(\rate)$ achieved by one-bit feedback ($\nolevel = 2$), respectively. Crosses and circles correspond to the rate points where the SNR-exponent of random codes does not achieve the optimal diversity.}
\label{fig:RDDT}
\end{center}
\end{figure*}
\newpage

\begin{figure}[tbf]
  \centering 
  \resizebox{1\textwidth}{!}{\input{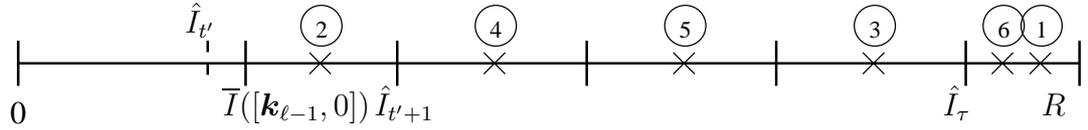}}
  \caption{An example of feedback threshold design ($\nolevel=12$).}
\label{fig:threshold_design}
\end{figure}

\begin{figure}[tbf]
\begin{center}
\includegraphics[width=1 \columnwidth]{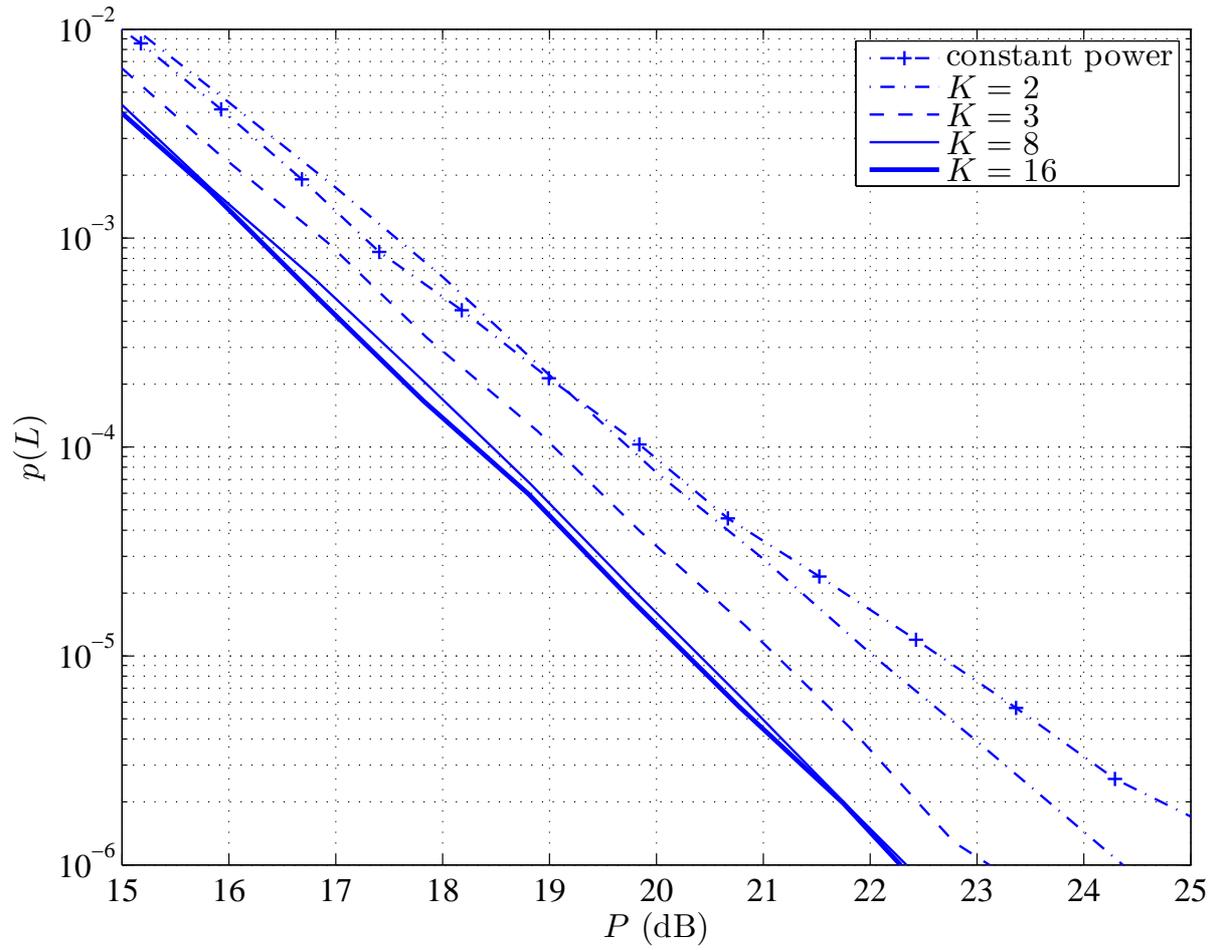}
\caption{Outage performance of  ARQ transmission schemes for a 16-QAM input block-fading channel with $\noround = 2, \notx=\norx=1, \noblock =2, \rate= 3.5$.}  
\label{fig:SISO}
\end{center}
\end{figure}
\newpage
\begin{figure}[tbf]
\begin{center}
\includegraphics[width=1 \columnwidth]{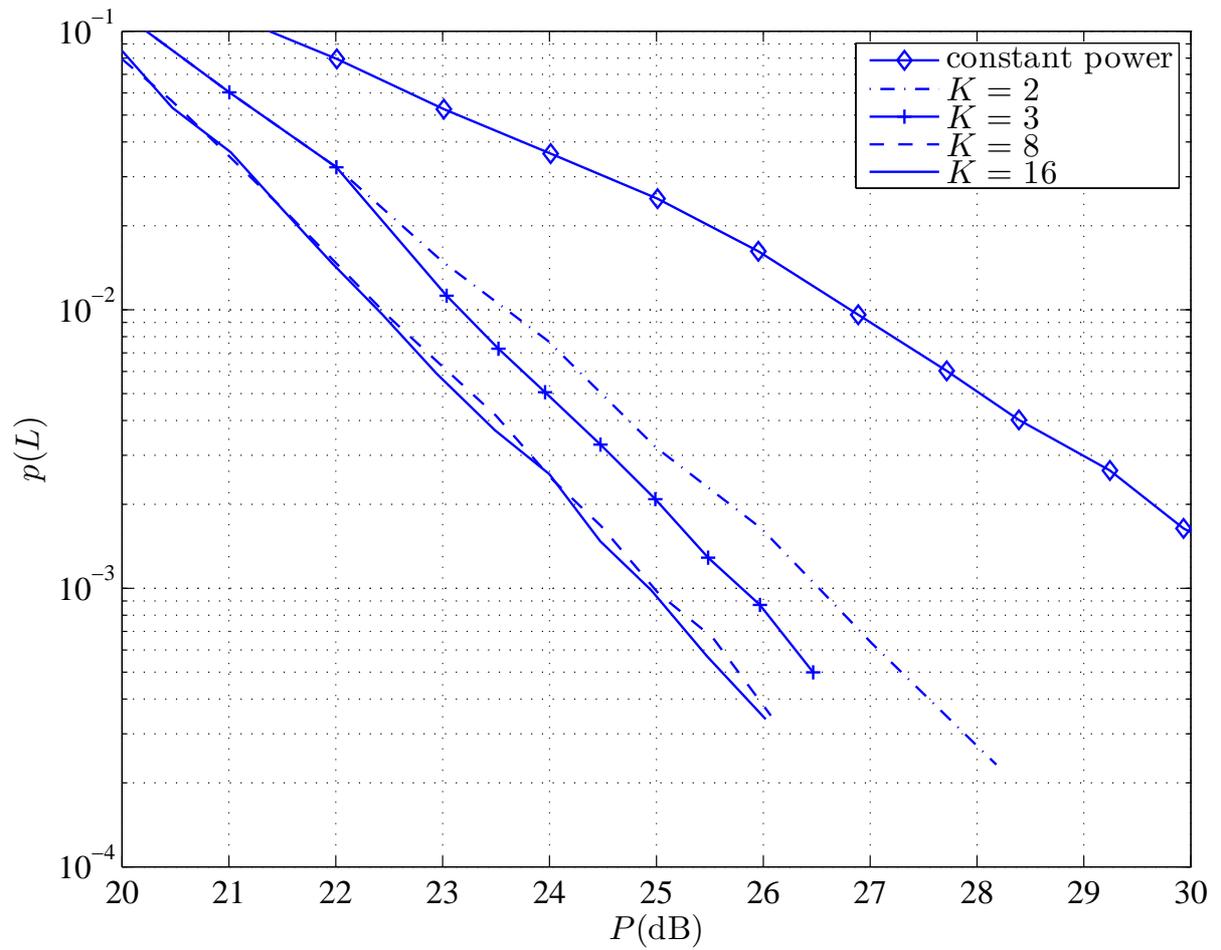}
\caption{Outage performance of  ARQ transmission  using the 16-QAM input constellation over the block-fading channel with $\noround = 2, \notx=2, \norx=1, \noblock =1, \rate= 7.5$.  Systems with constant transmit power, and systems employing power adaptation with $\nolevel=2, 3, 8, 16$ are considered.}  
\label{fig:MIMO_sim}
\end{center}
\end{figure}
\newpage
\begin{figure}[tbf]
\begin{center}
\includegraphics[width=1 \columnwidth]{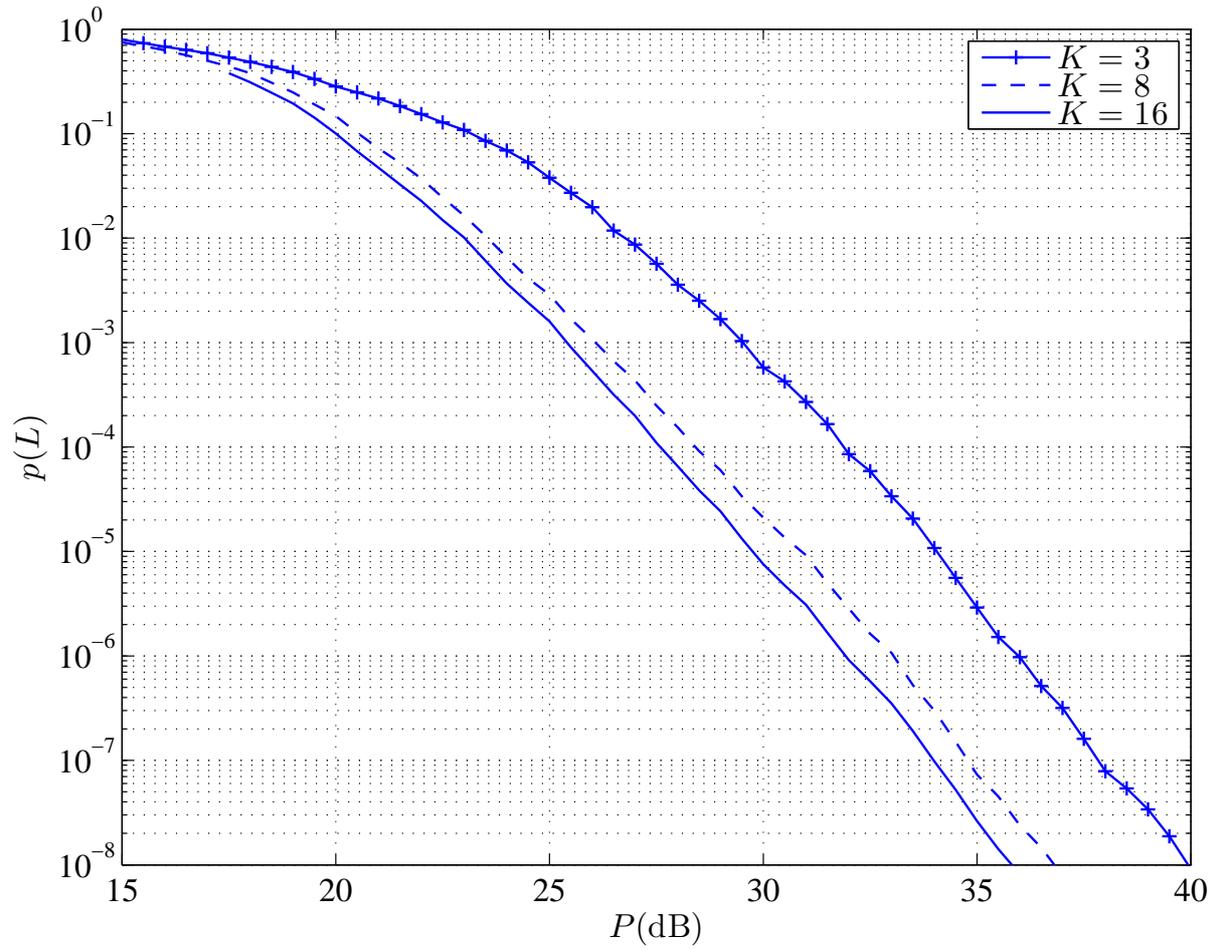}
\caption{Upper bound on outage performance of  ARQ transmission using 16-QAM input constellations over the  block-fading channel with $\noround = 2, \notx=2, \norx=1, \noblock =1, \rate= 7.5$ and $\nolevel=3, 8, 16$. }  
\label{fig:MIMO_bound}
\end{center}
\end{figure}

\end{document}